\documentclass{sig-alternate-10}

\newcommand{\ignore}[1]{} % ignore the part of text between { and }

\newcommand{\here}[1]{{\bf *** #1 ***}}

\newcommand{\SAT}{{\sf SAT}}

\newcommand{\MC}{{\sf MC}}
\newcommand{\MX}{{\sf MX}}
\newcommand{\MXE}{{\sf EV}}
\newcommand{\QE}{{\sf QE}}
\newcommand{\REACH}{{\sf REACH}$({\strA},\bigwedge \bar{\cal E} )$}
\newcommand{\Reachphi}{{\sf REACH}$_{_{\alpha}}({\strA}, \phi )$}
\newcommand{\TS}{{\sf TS}$_{\alpha}$}

\newcommand{\TMC}{{\sf temp-MC}}
\newcommand{\TSAT}{{\sf temp-SAT}}

\newcommand{\blue}[1]{\textcolor{blue}{#1}}

\newcommand{\red}[1]{\textcolor{red}{#1}}

\newcommand{\strA}{\cA}
\newcommand{\strB}{\cB}         % structure
\newcommand{\strC}{\cC}         %         % structure
\newcommand{\cA}{{\cal A}}
\newcommand{\cB}{{\cal B}}

\newcommand{\cC}{{\cal C}}

\makeatletter
\providecommand*{\cupdot}{%
  \mathbin{%
    \mathpalette\@cupdot{}%
  }%
}
\newcommand*{\@cupdot}[2]{%
  \ooalign{%
    $\m@th#1\cup$\cr
    \hidewidth$\m@th#1\cdot$\hidewidth
  }%
}
\makeatother

\ignore{
\usetikzlibrary{calc,shadows,arrows,decorations.pathreplacing}
\pgfdeclarelayer{background}
\pgfdeclarelayer{foreground}
\pgfsetlayers{background,main,foreground}
 
\tikzstyle{materia}=[draw, fill=blue!20, text width=6.0em, text centered,
  minimum height=1.5em,drop shadow]
\tikzstyle{module} = [materia, text width=8em, minimum width=10em,
  minimum height=3em, rounded corners, drop shadow]
\tikzstyle{texto} = [above, text width=6em, text centered]
\tikzstyle{linepart} = [draw, thick, color=black!50, -latex', dashed]
\tikzstyle{line} = [draw, thick, color=black!50, -latex']
\tikzstyle{ur}=[draw, text centered, minimum height=0.01em]

}

\newcommand{\x}{\bar{x}}                % tuple of variables

\newcommand{\bx}{\bar{x}}                % tuple of variables
\newcommand{\bX}{\bar{X}}  
\newcommand{\by}{\bar{y}}                % tuple of variables
\newcommand{\bY}{\bar{Y}}  
\newcommand{\ba}{\bar{a}} 
\newcommand{\bb}{\bar{b}}

\newcommand{\R}{\bar{R}}   

\newcommand{\bR}{\bar{R}} 

\usepackage{helvet}
\usepackage{paralist}
\usepackage{graphicx}
\usepackage{epsfig}
\usepackage{wrapfig}
\usepackage{tikz}
\usepackage{makeidx}
\usepackage{stmaryrd}

\usepackage{amsmath,amssymb,amsfonts,mathabx}

\usepackage{epic,eepic,epsfig,oldgerm,euscript}

\usepackage{xspace,enumerate,bm}

\newcounter{inlineequation}
\def\zd{,\ldots,}

\newtheorem{corollary}{Corollary}
\newtheorem{definition}{Definition}
\newtheorem{example}{Example}

\newtheorem{proposition}{Proposition}
\newtheorem{remark}{Remark}
\newtheorem{theorem}{Theorem}

\title{ Lifted Relational Algebra with Recursion \\
	and Connections to Modal Logic}
\author{Eugenia Ternovska}

\begin{document}

\maketitle

\begin{abstract}

	We propose a new formalism for specifying and reasoning about problems that involve heterogeneous ``pieces of information'' -- large collections of data, decision procedures of any kind and complexity and connections between them. The essence of our proposal is to lift Codd's relational algebra from operations on relational tables to operations on classes of structures (with recursion), and to add a direction of information propagation. We observe the presence of information propagation in several formalisms for efficient reasoning and  use it  to express unary negation and operations used in graph databases. We carefully analyze several reasoning tasks and  establish a precise connection between a generalized query evaluation and temporal logic model checking. Our development allows us to reveal a general correspondence between classical and modal logics and may shed a new light on the good computational properties of modal logics and related formalisms.

\end{abstract}

\section{Introduction}\label{sec:introduction}

Our goal is to develop a formalism for linking various pieces of information, including powerful solvers.
We take our inspiration in one the best-known and most influential formalisms in the Database community.

In 1970	Edgar (Ted) F. Codd introduced a relational data model
and two query languages:  relational calculus and relational algebra.
Relational calculus is what we usually call first-order logic.
The key contribution of Codd was to associate with a declarative specification language (first-order logic),
a  procedural counterpart which is the relational algebra, that later was implemented by smart engineers  and  became
	a multi-billion dollar industry of relational database
	management systems (RDBMS).

%amount of data, the radical shift in the amount of data, in the way that scientists distribute, 
%store, and aggregate this data, precipitates new challenges for knowledge representation

%scallability

A significant change has happened in the past decade.  
While, at the low level, everything data-related boils down to  SQL queries,
{\em interactions  between ``larger'', combinatorially harder, pieces became increasingly important.}
Such ``larger pieces'' are business enterprises, knowledge bases,
 web services, software, solvers in the world of declarative problem solving, 
 collections of learned data, potentially with numeric ranking assigned, etc.
In particular, combinations of powerful solvers is already  common in several communities, but has yet to make its way into the Database world.

While  database queries, expressed using Codd's  relational calculus, can be viewed 
as   {\em relations definable with respect to a structure} (a database),
{declarative problem specifications} can be understood as  
{ \em axiomatizations of classes of structures}.
The two notions 	 (in italic)	are defined in two consecutive chapters in the classic textbook 
on mathematical logic \cite{Enderton}.

%\subsection{ First Idea: Lifting Codd's Algebra}
Our {\bf first main idea} is to lift  Codd's algebra from operations on relational tables to operations 
on {\em classes of structures} (which we call modules), and to add recursion. 
Each atomic module can be given, e.g., by a set of constraints in a constraint formalism 
that has an associated solver, such as Answer Set Programming, or 
Constraint Satisfaction Problem, or Integer Linear Programming. It can even be an agent or a business enterprise,
a  collection of databases, or a database that is only partially  visible from the outside. 
Regardless of how modules are specified, our lifted relational algebra  views them as atomic entities.

%A module  can be given explicitly or implicitly, 
%e.g.  by a knowledge base,
%by an ontology, or by a constraint handling formalism

%  Satisfaction relation for each module is formalism-specific, 
%e.g., it can be Stable Model semantics 
%in case of ASP modules, and here it is abstractly represented by an interpretation that assigns
%a class of structures to a module.

%??? The first algebra is equivalent to first-order logic with least fixed point where instead of 
%predicate symbols we use atomic module symbols who's computational power we can control.

%\subsection{  Second  Idea: Adding Information Flow} 

Binary relations, or relations with some kind of directionality, are very common.
Problem solving often involves finding solutions for given inputs.
Most combinatorial problems are of that form, many software programs and hardware devices are of that form. 
Graph databases use binary relations to link data. Role specifications in Description Logics, which are used to describe ontologies, 
some constructs in Datalog$^+_-$, temporal connections and specifications of action effects have directionality or describe information propagation.
Our {\bf second main idea} is to add  
information flow to the algebra. By doing so, we force expressions with multiple variables into binary
and produce a calculus of binary relations. Modules with information propagation become (non-deterministic)
actions. Modules without information propagation become propositions.
Our  algebra with information flow is exactly like our first algebra, that is, like Codd's relational algebra (with  recursion),
but has information propagation added.

%Modules with information propagation  become 
%binary higher-order input-output relations or {\em non-deterministic actions}.
%Modules without information propagation become {\em propositions}.

%{\em Model expansion is the dominant task of all constraint solving technologies.}

%\subsection{Effects of Combining the Two Ideas}

The combination of the two ideas allows us to { demonstrate a general connection between classical and modal logics}
through investigating the effects of adding information flow. 	
We show that many formalisms for efficient reasoning are instances of the same phenomenon,
which is also responsible for good computational properties of modal logics. 
Through our detailed study of several reasoning tasks, we provide an explanation of robust decidability 
of modal logics.		
		
	%The framework  addresses representational and computational scalability.	
\ignore{		In particular, because of lifting, we can uniformly talk about states as structures (databases).
			Interpreting our algebra over a transition system
		gives rise to a modal temporal logic. 
			In doing so, we discovered ...
		It turns out that this tiny addition produces 
		is responsible for phenomena 	
	}

\subsubsection*{Structure of the Paper}

%{Algebra} 
In Section \ref{sec:Algebra}, %starting from classical  logic with a fixed point, we develop an algebraic view on module combinations
we  define our lifted version of Relational Algebra with recursion.
We call this version of the algebra ``flat'' to distinguish it from the version 
with information propagation.	We use a model-theoretic semantics, although various generalizations are possible.
%{Model Expansion,  Related Tasks} 
In Section \ref{sec:MX-Related-Tasks}, we define several reasoning tasks, in particular, Model Expansion that is 
responsible for adding information propagation.

In Section \ref{sec:Algebra-Inf-Flow}, we define our algebra with  information flow, that we also  call Dynamic Algebra  or a Calculus of Binary Relations.
In this algebra, second-order variables provide an easy way to model data flow.
We interpret the Dynamic Algebra  
over transition systems where  states are structures.
We show that many interesting operations such as unary negation, constant tests and subexpression tests 
are definable in the algebra. We also discuss sequential composition and the reverse operations.
Using the  Dynamic Algebra, we produce a modal temporal logic $L\mu\mu$
%that can act as a high-level programming language,
and prove the equivalence of this modal logic to the calculus of binary relations.
We show that several formalisms, an expressive Description Logic, Dynamic Logic, and Datalog$^+_-$ 
also use information propagation and can  be viewed as fragments of $L\mu\mu$. 

In Section \ref{sec:Queries-Machines-Modalities}, that we call {\em Queries, Machines, Modalities}, 
we introduce a model of computation (simple abstract machines) that are suitable for declarative problem solving.
The machines are similar to Abstract State Machines of Yury Gurevich, however, in our machines, not only states, but also  actions are structures.
We discuss several important complexity measures. Formulae of our modal logic are programs for these machines. 
We describe reaching a solution to an algebraic specification in terms of these new devices.
We then define a generalization of the Query Evaluation problem.
Our main result is an  equivalence  between 
the reachability in the execution graph, the general Evaluation problem and and temporal model checking. 
 The property holds under any assignment of the direction of information propagation to the internal 
 relational variables. We then provide our explanation of why modal logics are so robustly decidable.

\section{``Flat'' Algebra}\label{sec:Algebra}
We call this version of the algebra ``flat'' to distinguish it from the version 
with information propagation below.

\subsection{Syntax of the ``Flat'' Algebra}

Let $\tau_M =\{ M_1, M_2, \dots \}$ be a fixed vocabulary of {\em atomic module symbols}.
Atomic module symbols are of the form $M_i(X_{i_1},\dots,X_{i_k})$,  (also written $M_i(\bar{X})$), 
where each $X_i$ is a relational variable.  Each $X_j$ has an associated arity $a_j$. The set $\{ X_{i_1},\dots,X_{i_k} \}$  is called the {\em variable vocabulary of $M_i$} and is denoted $vvoc(M_i)$.
%\blue{We also allow  relational constants from $\tau$ in place of the variables, provided their arities match.} In this case, $vocab(M_i)$ denotes the combined
%(variable and constant) vocabularies of $M_i$.
Modules in $\tau_M$ are {\em atomic}. Modules that are not atomic are called {\em compound}.
Algebraic expressions for {\em modules} are built by the grammar:
	\begin{equation}\label{eq:algebra}
    E::= \\ \bot |  M_i  | Z_j  | E \cup  E | 
- E | \pi_{\delta} E | \sigma_\Theta E |  \mu Z_j. E.
\end{equation}
Here,   $Z_j$ is a \emph{ module variable}. It must occur positively in the expression $E$, i.e., under an even number of the complementation ($-$) operator.
 The  set-theoretic operations are union ($ \cup$) and complementation ($-$). Intersection ($ \cap$) and set difference are expressible. 
Projection ($\pi_\delta E$) is a family of unary operations, one for each possible set of relational variables $\delta$. Each  symbol  in $\delta$ must appear in $E$. The condition  $\Theta$ in selection $ \sigma_{\Theta} E $ is an expression of the form $L_1\equiv L_2$, where $L_i$ is a relational variable or `$R$', where $R$ is a relation (set of tuples of domain elements).\footnote{A more general  version allows $\Theta$ to be built 
	up using Boolean operations from the
	equivalence and non-equivalence operators, $ \equiv$,  $\not \equiv$. That choice of $\Theta$ may be more convenient to implement,
	but does not add expressive power since it is expressible trough the other operations.}
Thus, we bring semantic elements into syntax, and they became constants in the language.   
%\footnote{Selection can be used, in particular, to connect modules  
%	by equating relational symbols of equal arity, and to express {\em grounding} by brining relations over domain elements into the syntax.}
%\red{I may drop this part about brining semantic elements into the syntax, it is too specialised.}
The operations (except  $\mu Z_j. E$) are essentially like in Codd's relational algebra, except we include full complementation ($-$) instead of set difference for generality. However, the operations are on objects of  a higher order -- on {\em classes of structures} rather than on relational tables.

\begin{remark}
While we follow an algebraic approach in presenting the syntax, it will be seen from the semantics that  {\bf the  constructs in this paper  work the same way  as the corresponding constructs in logic. } 
\end{remark}
The formalism is equivalent to a ``lifted'' version of first-order logic with the least fixed point operator (FO(LFP)).
The ``lifting'' means that instead of regular predicate symbols, we have modules who's computational power we can control, and instead of object variables that range over domain elements, we have second-order variables ranging over relations. 
In particular, projection is onto a set of {\em relational} variables, selection is equality of relations rather than of domain elements.
%Thus, our formalism is a {\bf lifted relational algebra with recursion}.
Note that first-order quantification is ``encapsulated'' in an axiomazation of each atomic module (e.g., an ILP or an ASP module) and is not visible from the outside of that module.
Another way of viewing our formalism is SO$+$LFP$-$FO, that is, second-order logic with least fixed point ``without'' the first-order part.

\subsection{Examples}

We now give several examples of combinations of modules (pieces of information)  in the algebra. 
For simplicity, we  use common combinatorial problems  since they are very familiar to most readers. 
	A combinatorial problem can be viewed as %a knowledge base which specifies a module, i.e.,
 a class of structures.\footnote{A module is a {\em set} of structures when the domain is fixed.}

\begin{example}\label{ex:HC-2Col} 
	{\rm  
		Let $M_{\rm HC}(V,X,Y)$ and $M_{\rm 2Col}(V,X,Z,T)$ be atomic modules ``computing'' a Hamiltonian Circuit and a 2-Colourability. They can do it in different ways.
		For example, $M_{\rm HC}$ can be an Answer Set Programming program, and $M_{\rm 2Col}$ be an imperative program or a human child with two pencils.  Here, $V$ is a relational variable of arity 1, $X,Y$ are relational variables of arity 2, and the first module decides if $Y$ forms a Hamiltonian Circuit (represented as a set of edges) in the graph given by vertex set $V$ and edge set $X$. Variable $X$ of the second module has arity 2, and variables $Z,T$ are unary; the module decides if unary relations $Z,T$ specify a proper 2-colouring of the graph with edge set $X$. The following algebraic expression determines a combination of 2-Colouring and Hamiltonian Circuit, that is whether or not there is a 2-colourable Hamiltonian Circuit.\footnote{We use $:=$ for ``is by definition''.}
		\begin{equation} \label{3-Col-HC}
		\begin{array}{l}
		M_{\rm 2Col-HC}( V,X,Z,T) \ :=\  \\
		%%  \hspace{-5mm}	\pi_{E_1,Col} (\sigma_{C_V\equiv E_2)} (M_{\rm HC}(E_1, C_V)  \cap 	M_{\rm 3Col}(E_2,Col)).
		\hspace{-5mm}	\pi_{V, X,Z,T} ((M_{\rm HC}(V, X, Y)  \cap 	M_{\rm 2Col}(V, Y,Z,T)). 	
		\end{array}
		\end{equation}
		%% Here,  the selection  operator requires the interpretations of $C_V$ in $M_{\rm HC}$ and of $E$ in $M_{\rm 3Col}$ to be the same.
		%% Thus, Hamiltonian Circuit gets coloured. 

		Projection ``keeps'' $V,X,Z,T$ and  hides   the interpretation of $Y$ in  $M_{\rm HC}$, since it is the same as $Y$'s in $M_{\rm 2Col}$.
	}
\end{example}

\begin{example}        \label{ex:LSP}   
	{\rm  	This modular system can be used by a company that provides logistics services 	(arguments of atomic modules are omitted).
		$$M_{LSP}:= \sigma_{B\equiv B'}(M_K  \cap M_{TSP}).$$
		It decides how to pack goods and
		deliver them.
		It solves two NP-complete tasks interactively, -- Multiple Knapsack (module $M_K$) and Travelling Salesman Problem (module $M_{TSP}$).
		The system takes orders from customers (items  to deliver, their profits, weights), and the capacity of trucks, decides how to pack  items in trucks, and for each truck, 
		solves a TSP problem. 
		The feedback  $B'$ about solvability of TSP is sent back to $M_K$.
		%	Module $M_{TSP}$ takes a candidate solution from $M_K$, together with the graph of cities and routes with distances, allowable distance limit and destinations for each product. It computes the route for each truck. 
		%	The Knapsack problem is solved, 
		% e.g. using Integer Linear Programming (ILP), and TSP using Answer Set Programming (ASP), and they are 
		The two sub-problems, $M_K$ and $M_{TSP}$, are solved by different sub-divisions of the company (potentially, with their own business secrets) that cooperate towards the common goal. 
		%	The modules $M_K$ and $M_{TSP}$ are composed in sequence, with a feedback going from an output 
		%	of $M_{TSP}$ to an input of $M_K$.  
		A solution to the compound module, $M_{LSP}$, to be acceptable, must satisfy both sub-systems.
	}
\end{example}

%Many practical examples use simple combinations of modules, where only conjunctions, disjunctions, perhaps with projections and selections, are used.
%The university timetabling example \cite{JOJN} and Examples \ref{ex:HC-2Col}, \ref{ex:LSP} are of this kind. 
In some specifications, the use of a {\bf recursive construct} is essential. For example, we may need to specify a recursive algorithm,
or the semantics of the satisfaction relation in a logic, which is given by an inductive definition.
%In the following example, recursion is used to specify an algorithm for solving a combinatorial search problem on graphs using  tree decompositions. 
\begin{example}{\rm 
		In this example, we need recursion to specify a Dynamic Programming algorithm on a tree decomposition of a graph.
		The modules we use are $M_{\rm TD}$ that performs tree decomposition of each graph, and
		$M_{\rm 3Col}$ that is used recursively to performs 3-Colouring on each bag of the decomposition. 
		The problem is represented as 
		%\begin{small}
			$$M_{\rm TD}  \cap \mu Z. \Psi ( Z, M_{\rm TD},M_{\rm 3Col}  ),$$ 
		%\end{small}
		\noindent where  $Z$ is a module variable over which recursive iteration is performed, the least fixed point expression
		$\mu Z. \Psi ( Z, M_{\rm TD},M_{\rm 3Col}  ) $ specifies the dynamic programming algorithm.
		Details are omitted because of lack of space.
}		
\end{example}

\subsection{Valuations}

To interpret   algebraic expressions, we use valuations. % $(v, {\cal V})$.
This is an important notion used throughout the paper.
\begin{definition}[Valuation] Valuation $(v, {\cal V})$ is a pair of functions. Function
 $v$ maps relational variables in $vvoc(M_i)$ to symbols in a relational  vocabulary $\tau$ so that the arities of 
the relational variables in $vvoc(M_i)$  match those of the corresponding symbols in $\tau$. 
Function $\cal V$ is parameterized by $v$ and provides a domain  (which does not have to be finite),
and  interpretations of atomic modules $M_i$ as follows. 
Let $V$ be the set of all  $\tau$-structures over the domain fixed by $\cal V$.
Valuation $\cal V$ maps each atomic module symbol $M_i$ to a subset ${\cal V}(v,M_i)$ of  $V$
so that for any two $\tau$-structures $\cA_1$, $\cA_2$ which coincide on $vocab(M_i)$,
we have  $\cA_1 \in {\cal V}(v,M_i)$ iff  $\cA_2 \in {\cal V}(v,M_i)$. 
\end{definition}
All of the above applies to module variables $Z_j$ as well.
In practice,  $\cal V$ can, for example,  associate one module symbols with stable models of an ASP 
program, another module symbol with models of an ILP encoding, yet another one with a set of databases used by a particular enterprise, etc.
We will also see, in part \ref{subsubsec:SO-FO}, that it is also possible to treat each modules simply as a predicate symbol.     

%which we call an {\em extension} of $M_i$. Extensions for all well-formed algebraic expression are defined inductively.
\ignore{, where $v$ is a predicate valuation
	that maps each predicate variable in $vvoc(M_i)$ to a predicate constant (of the same arity) from $\tau$; and $\cal V$ 
	is a module  valuation parameterized by $v$. First, module valuation $\cal V$  fixes the domain. The domain does not have to be finite.
	Let ${\cal C} $ be the set of all  $\tau$-structures over this domain.
	Modules (atomic and compound) are interpreted by subsets of  $V$.
	Module valuation $\cal V$ maps a  or an atomic module symbol $M$ to a subset ${\cal V}(v,Z)$ 
	(respectively, ${\cal V}(v,M)$) of $\cal C$ such that for any two $\tau$-structures $\cA_1$, $\cA_2$ which coincide on $vocab(M_i)$,
	we have  $\cA_1 \in {\cal V}(v,M_i)$ iff  $\cA_2 \in {\cal V}(v,M_i)$. 
	The same condition applies to module variables $Z_i$.
	}

\begin{remark} %Note that  $v$ resembles a	``call by reference'' in programming.
	 Valuations $\cal V$ (parameterized with $v$) can be viewed as ``oracles'' or decision procedures associated with modules,
	and can be of arbitrary computational complexity.
\end{remark}

\subsection{Semantics of the ``Flat'' Algebra}

%{\em Interpretation} ${\cal I} (M_i)$ of module $M_i$ is a set of $\tau$-structures 
%the variables of a module $M(X_1\zd X_k)$ are interpreted as predicates over some domain and the module `evaluates' them.
The {\em extensions} $\llbracket E \rrbracket ^{{\cal V},v}$ of algebraic expressions $E$ are subsets of $V$ (the set of all  $\tau$-structures over the domain fixed by $\cal V$)
defined as follows.
$$
\begin{array}{l}
\llbracket \bot \rrbracket ^{{\cal V},v} : =  \varnothing.\\
\llbracket M_i\rrbracket ^{{\cal V},v} : = {\cal V} (v, M_i) \ \mbox{for some}\  v.\\
% \{  \cA  \  \ |\  \cA \models_{\cal I}  M_i     \}.\\

\llbracket Z_j\rrbracket ^{{\cal V},v} : = {\cal V} (v, Z_j) \ \mbox{for some}\  v.\\

\llbracket E_1  \cup  E_2\rrbracket ^{{\cal V},v} : = \llbracket E_1\rrbracket ^{{\cal V},v} \cup  \llbracket E_2\rrbracket^{{\cal V},v}.\\

%\llbracket E_1  \cap  E_2\rrbracket ^{{\cal V},v} : = \llbracket E_1\rrbracket ^{{\cal V},v} \cap  \llbracket E_2\rrbracket^{{\cal V},v}.\\

\llbracket  -E \rrbracket^{{\cal V},v}  : =   V \setminus  \llbracket E\rrbracket ^{{\cal V},v}.\\

\llbracket \pi_\delta(E) \rrbracket ^{{\cal V},v} : = \{   \cA  \  \ |\   \exists {\cA}' \ (
\cA' \in \llbracket E\rrbracket ^{{\cal V},v} \mbox{ and }  {\cA}|_\delta={\cA}'|_\delta \      ) \}.\\

\llbracket \sigma_{L_1\equiv L_2} E \rrbracket ^{{\cal V},v} : = \{   \cA  \  \ |\  \llbracket E\rrbracket^{{\cal V},v}\  \mbox { and } \ L_1^{\cA} = L_2^{\cA} \}.\\

\llbracket \mu Z_j.E \rrbracket ^{{\cal V},v}  : = \bigcap \big\{ {\cal E} \subseteq  {\cal C}  \ | \ \llbracket E\rrbracket ^{{\cal V} [Z:={\cal E}],v}\subseteq {\cal E} \big\} .\\
\end{array}	
$$
Here, ${\cal V} [Z  {:=} {\cal E} ]$ means a valuation that is exactly like  ${\cal V}$ except $Z$ is interpreted as $\cal E$.
%We say that $\cA$ is a \emph{model} of $M$ in ${\cal I}$ with variable assignment $v$ (notation $\cA \models_{\cal I}M[v]$) if $\cA \in {\cal I}^v(M)$.
Note that projection restricts each structure $\cB$ of $M$ to  $\cB|_\delta$ leaving the interpretation of other symbols open. Thus, it {\em increases} the number of models. Selection {\em reduces}  the number of models.

This algebra may look very different from Codd's algebra because all modules are sets of $\tau$-structures, that is, sets of tuples 
of the same length, while in Codd's algebra the length of tuples varies. 
There is no contradiction here --  Codd's algebra is just what is seen through the ``window of variables''.

\begin{definition}[Satisfaction, ``flat'' algebra]
	Given a well-formed algebraic expression $E$ defined by (\ref{eq:algebra}), we say that structure $\cal A$ {\em satisfies} $E$ under 
valuation $({\cal V},v)$, notation 
$${\cal A} \models_{({\cal V},v)}  E,\ $$
 if ${\cal A} \in \llbracket E \rrbracket ^{{\cal V},v}$. 
\end{definition}

%Intuitively, modules can be understood as procedures, and $({\cal V},v)$ help to represent call by reference.

\ignore{ Intuitively, to check whether a $\tau$-structure $\cA$ satisfies module $M_i$, we need to first select predicate symbols $S_{i_1}\zd S_{i_k}$ from $\tau$, whose arities match those of $X_1\zd X_k$, which is done by function $v$, and then ``apply'' the module to $S^\cA_{i_1}\zd S^\cA_{i_k}$,
 as we would apply a decision procedure. 
}

\begin{remark}
	Note  that while individual modules are already capable of solving optimization tasks (the optimum value can be given 
	as an output in one of the arguments),
	the least fixed point construct can generate the least value over a collection of modules combined in an algebraic expression.
\end{remark}

\ignore{\here{In proposition 1 it is mentioned that “-“ represents negation. However, various references to negation are done throughout the paper, and it would be good to understand right upfront which kind of negation is actually tackled.
}	}

\subsection{Representation in Logic}

 \begin{proposition}[Logic Counterpart]
 	The  algebraic  operations are equivalently representable in logic, where `$ \cup$' corresponds to 
 	disjunction, `$-$' to negation, `$\pi_{\nu}$' to second-order existential quantification over $\tau \setminus \nu$, `$ \sigma_{\Theta}$' to conjunction with $ \Theta$, $\mu Z.E $ to the least fixed point construct.
 \end{proposition}
%\blue{ Thus, our formalism is SO(LFP) over modules that are of an arbitrary expressive power. {\tt Clarify! }}
 %\begin{notation}
 %	We will use $ \sigma_{\Theta}$ even in the logic setting to separate  $ \Theta$ from the rest of the formula.  ??????????????
% \end{notation}
 \begin{example} 
 	{\rm 
 	Expression (\ref{3-Col-HC}) for $M_{\rm 2Col-HC}(V,X,Z,T) $ is represented in logic  as  
 	$$
 	\begin{array}{l}
 	%\label{eq:logic}
 	%% \hspace{-5mm}	 \exists E_2  \exists C_V  [ \sigma_{(C_V\equiv E_2)} [M_{\rm HC}(E_1, C_V) \land 	M_{\rm 3Col}(E_2,Col)]],	
 	\hspace{-5mm}	 \exists Y  [M_{\rm HC}(V,X,Y) \land 	M_{\rm 2Col}(V,Y,Z,T)].	 \end{array}
 	$$
 	%% or simply as
 	%% $
 	%% \begin{array}{l}
 	%% \label{eq:logic}
 	%% 	  \exists E_2  \ \ [  M_{\rm HC}( E_1, E_2) \land 	M_{\rm 3Col}(E_2,Col)  ].
 	%% \end{array}
 	%% $

}	
 	
 \end{example}

 Note that first-order variables can  be mimicked with second-order variables over singleton sets.
 \begin{proposition}
 	When all relations are unary and the sets that interpret them are are singletons,  the formalism collapses to  FO(LFP).
 \end{proposition}

\section{Model Expansion,  Related Tasks} \label{sec:MX-Related-Tasks}

Model expansion \cite{MT05-long} is the task of expanding a structure to satisfy a specification (a formula in some logic).

\ignore{
The reader who wants to get to our main result as quickly as possible should now read subsection \ref{subsec:BinaryRelations} and go directly to Section \ref{sec:Algebra-Inf-Flow} and then \ref{sec:ModalRobustDecidability}. The reader who wants to understand our claim that the satisfiability problem is overrated, in particular in the attempts to explain the robust decidability of modal logics, should read the  subsections in detail.
}

\subsection{ Three Related Tasks: Definitions}
For a formula $\phi$ in any logic $\cal L$ with model-theoretic semantics, we can associate the following three tasks
(all three for the same formula), satisfiability (\SAT), model checking (\MC) and model expansion  (\MX).
We now define them for the case where $\phi$ has no free object variables.

\begin{definition}[{\sf Satisfiability (SAT$_\phi$)}]
	\underline{Given:} \\ Formula $\phi$.
	\underline{Find:} structure $\strB$
	such that  $\strB\models \phi$. (The decision version is: \underline{Decide:} $\exists \strB$ s.t. $\strB\models \phi$?)
	
\end{definition}

%In a modal (temporal) logic setting, this task is usually associated with the non-emptiness problem
%for the corresponding automaton.

\begin{definition}[{\sf Model Checking (\MC)}] \label{def:MC}
		\underline{Given:}\\ Formula $\phi$, structure $\strA$ for $vocab(\phi)$.
		\underline{Decide:}  $\strA \models \phi$?
	\end{definition}
There is no search counterpart for this task. 
%Temporal logic model checking is different,
%and will be discussed separately.

The following task (introduced in  \cite{MT05-long}) is at the core of this paper.
The decision version of it is of the form ``guess and check'', where the ``check'' part 
is the model checking task we just defined. 
\begin{definition}[{\sf Model Expansion (MX$_\phi^{\sigma}$)}] \label{def:MX}
	\underline{Given:} \\ Formula  $\phi$ with designated input vocabulary $\sigma\subseteq vocab(\phi)$ and
	$\sigma$-structure $\strA$.
	\underline{Find:} structure $\strB$ such that 
	 $\strB\models \phi$ and expands
		$\sigma$-structure $\strA$ 	to $vocab(\phi)$.
	(The decision version is: \underline{Decide:} $\exists \strB$ such that  $\strB\models \phi$ and expands $\sigma$-structure $\strA$ 
		to $vocab(\phi)$?)
\end{definition}
Vocabulary $\sigma$ can be empty, in which case the input structure $\strA$ consists of a domain only. 
When $\sigma = vocab(\phi)$, model expansion collapses to model checking, $\mbox{\sf MX}_\phi^{\sigma} = \mbox{\sf MC}_\phi$.
Note that, in general, the domain of the input structure in \MC{\ }and \MX{\ }can be infinite. 
%The computational  complexity of the task can be studied in terms of the size of the finite active domain.

 Let $\phi$ be a sentence,
i.e., has no free object variables.
Data complexity \cite{Var82} is measured in terms of the size of the finite active domain. Data complexity of  \MX{\ }lies in-between model checking (full structure is given) and satisfiability (no structure is given).
$$ \mbox{\sf MC}_\phi \leq \mbox{\sf MX}_\phi^{\sigma} \leq \mbox{\sf SAT}_\phi. 
$$
Of course, we consider the decision versions of the problems here. For example, for FO logic, \MC {\ }  %is boolean query evaluation and 
is non-uniform AC$^0$, \MX {\ }  captures NP (Fagin's theorem),
and \SAT  {\ }  is undecidable. In  \SAT, the domain is not given. In \MC {\ } and \MX,  the (active) domain is always given, which significantly reduces 
the complexity of these tasks compared to \SAT.
The relative complexity of the three tasks for several  logics  has been studied in \cite{Kolokolova:complexity:LPAR:long}.

\ignore{
{\em Model expansion for an expressive logic  and satisfiability for a corresponding (and much less expressive) logic are connected in the same way as Fagin's and Cook's theorems are connected (see, e.g. page ??? of \cite{Libkin04:book-long}), through a step that in practical systems is called {\em grounding}, that amounts to instantiating variables with constants that correspond to the domain of the instance structure. }
For example, this step is implemented in ASP systems e.g. Clasp \cite{clasp}. \blue{Enfragmo} \footnote{In the ASP literature, inputs are given as sets of ground atoms,
	and the semantics is discussed in terms of Herbrand structures, which obscures the separation of problem instances (structures, databases, a semantical notion here) from problem descriptions (axiomatizations of combinatorial problems).} 
}

\subsection{Similar Task: Query Evaluation}
In database literature, a task similar to \MX is query evaluation.
We define it for non-boolean queries, i.e., those with free object variables $\by$.

\begin{definition}[{\sf Query Evaluation (\QE)}]\label{def:QE} {\  } \\
	\underline{Given:} Formula  $\phi(\bar{x})$ with free object variables  $\bar{x}$, 
	%each 	ranging over domain elements, 
	tuple of domain elements $\bar{a}$ and
	$\sigma$-structure $\strA$, where $\sigma= vocab(\phi)$.
	\underline{Decide:}  $\strA \models \phi[{\ba}/\bx]$? 
\end{definition}

The task specifies the
combined complexity of query evaluation. 
Data complexity  requires formula $\phi$ to be fixed. To analyze expression complexity, we 
fix the database \cite{Var82}.

For {\em the same logic}, these two tasks, \MX{\ }and \QE, (as we defined them)  have very different data complexity,
e.g., for first-order logic, query evaluation is in AC$^0$, model expansion is NP-complete. %For this reason, we prefer to treat them separately.
%Model expansion corresponds to a non-deterministic action in the dynamic setting, while QE is deterministic.

Database query $\phi$ can be viewed 
as  a {\em relation definable with respect to a structure} (a database). See, e.g., the classic textbook by Enderton \cite{Enderton}.
The action of defining such a relation is always deterministic -- it defines one relation.
This relation is a set of tuples, it is not a part of the vocabulary.

\subsection{Model Expansion  as a Binary Relation}\label{subsec:BinaryRelations}

When we talk about powerful solvers (such as those that allow us to compute, say, 3-Colourability or to come up with solutions to logistics problems),
we usually produce multiple solutions. In data graphs, we may have multiple children of the same data node. 
Later in this paper, we will talk about a transition system with non-deterministic actions and a calculus of binary 
relations. Model expansion (as in Definition \ref{def:MX}, the Find part) gives us  such a binary relations on structures. 

\begin{example} {\rm 
		 3-Colourability is a binary relation such that $(\strA,\strB)$ is in this relation if and only if $\strA$ contains an interpretation of the edge and vertex relations (that is, a particular graph), and  $\strB$ contains an interpretation of relational symbols $R$, $B$, $G$ that represents a particular correct colouring of this graph.  Note that both $\strA$ and $\strB$ may also interpret other relational symbols, but those interpretations do not matter for  this relation. Those extra things may be called garbage of the computation.}
	\end{example}

\subsection{The  Four Tasks in Applications} % (\MC, \MX, \SAT, \QE)

In database research, \QE{\  }has already been  studied extensively, e.g.   for first-order logic (Codd's relational algebra) 
and its fragments (e.g. for  conjunctive queries), as well as for DATALOG and its variants.
\SAT{\  }has demonstrated its practical importance mostly for propositional logic, a logic of a very low expressive power.
Indeed,  the success in SAT solving (achieved both by smart algorithms and an exponential growth in hardware capabilities)
is one of the most remarkable achievements of logic in Computer Science. Due to this success, the complexity class NP
is often called ``the new tractable''. However beyond the propositional case, the great majority of logics that are interesting and useful in Computer Science 
are undecidable. For instance, first-order logic is undecidable even in the finite (by the Trakhtenbrot's theorem).
In addition, integration of theories often presents as a major problem in Knowledge Representation and in  Satisfiability Modulo Theory because a combination of 
two theories is often undecidable. However, in practice, a finite domain is often given on the input, and, in such a case, the undecidability problem
for combinations of theories 
does not arise. 

% and, at the same time,  success stories, systems such as ASP-CP are cited\cite{KR-worshop-report}.
% Such systems 
 %combine theories, but do not encounter the 
 
 Moreover, systems for logics with a very high complexity of satisfiability, often  perform very well in practice. 
 The explanation is that  those systems solve \MX, not \SAT, since an (active) domain is a part of the input. 
%When the domain is given on the input, undecidability in those systems is not an issue. 
While \SAT{\ }continues to be important for propositional logic, the importance of this task for expressive  logics used in practice
is {\em greatly overrated}. 
%On the other hand,  {\em unlike satisfiability, model expansion (along with query evaluation) for expressive logics 
%	has much better computational properties} \cite{Kolokolova:complexity:LPAR:long}. % and is very common in practise.
Just as the query evaluation problem is prevalent in database research, model expansion is very common in the general area of constraint solving.
Most constraint solving paradigms solve \MX{\ }as the main task, e.g. 
%\red{ REPLACE WITH CSP AND DB EXAMPLES LOGICBLOX in AI planning, scheduling,
logistics, supply chain management, etc.
Java programs, if they are of input-output type, can be viewed as 
model expansion tasks, regardless of what they do internally. Most combinatorial problems are of that form, many software programs (e.g., web services) and hardware devices (e.g., circuits) are of that form.  The Logistics Service Provider in Example  \ref{ex:LSP} has, e.g., customer requests
as an input, and routes and packing solutions as  outputs.
In Example \ref{ex:HC-2Col}, one can 
have e.g. edges of a graph on the input 
to formula  (\ref{3-Col-HC}), and colours on the output.  ASP systems, e.g., Clasp \cite{clasp} mostly solve model expansion, 
and so do CP languages such as Essence \cite{Essence}, as shown in \cite{MT-essence-journal:2008}. Problems solved in ASP competitions are mostly 
in model expansion form.
CSP in the traditional AI form (respectively, in the homomorphism form) is representable by model expansion where mappings to domain elements (respectively, homomorphism functions) are expansion functions. 
%Several examples are given in the \blue{Appendix ?????????}.
%\footnote{\red{Talking about these systems answering second-order queries 
%		would be mathematically correct, but impractical since their languages are often very far from classical logic, and may be even imperative.???????????}}

In the algebra we  present next, we can view each algebraic expressions  as a 
{\em network of inter-connected solvers and databases}, 
jointly solving one task:
satisfy all the components. 
	We deal with two types of objects: 
\begin{compactitem}
	\item Static Objects:  model checking modules $M_p$:
	\begin{compactitem}
		\item collections of databases 
		\item 	{\em decision} procedures of any kind and complexity, e.g. is a given graph 3-colourable?
		\item relations or collections of objects
	\end{compactitem}
	\item Dynamic Objects: model expansion modules $M_a$:
	\begin{compactitem}
		\item actions, changes 
		\item combinatorially complex {\em  search} and {\em optimization} problems, e.g. planning, scheduling, TSP
		\item any binary relations on (sets of) structures
		\item roles  in Description Logics 
		\item data links in graph databases
		\item causality links
	\end{compactitem}
\end{compactitem}

\section{Algebra with  Information Flow}\label{sec:Algebra-Inf-Flow}

We now describe a transformation from the  ``flat'' algebra to the ``dynamic'' one, which gives rise to a modal logic. 
All we do is we {\bf add information propagation}.
Some atomic modules serve as propositions. They are unchanged by the transformation. Other become actions.
In each atomic action-module $M$, we \underline{underline} designated input symbols and denote them $\sigma_M$. Output symbols
are  those that are free (are not quantified) in the algebraic expression where $M$ occurs. They are denoted $\varepsilon_M$.
Thus, {\bf we  force  multi-dimensional expressions into binary}. For compound expressions $\alpha$, we use $\sigma_{\alpha}$ to denote
the union of all  input symbols that occur free, and $\sigma_{\varepsilon}$ to denote all output symbols of $\alpha$ that occur free.
\begin{example}
	{\rm 
	Consider again the  HC-2Col example:
	\begin{equation} \label{HC-2Col-MX}
	\exists Y  [M_{\rm HC}(\underline{V},\underline{X},Y) \land 	M_{\rm 2Col}(\underline{V},Y,Z,T)].
	\end{equation}
	The quantified symbol $Y$  is not visible from the outside. 	The output vocabulary of this compound modular system is
	$\varepsilon = \{Z,T\}$ (for the two colours),
	the input vocabulary is $\sigma = \{V,X\}$ (for the  vertices and edges). 
	In general, {\em any} direction	of information propagation can be specified. 
	For the external (free) symbols, a particular specification of inputs and outputs determines which problem 
	we are solving. For the internal symbols (those that are quantified), it does not matter which symbols are inputs, which ones are outputs.
	For instance, the internal symbol $Y$  can be considered as an input to the 
	second module, or it can be a symbol on the output who's value is guessed and checked to satisfy both modules.

	%, e.g. from colours to graphs in 2-Colouring.
	%When selection specifies the interpretations of all predicates, we have model checking task.	
}
\end{example}

\subsection{Minimal Syntax of the Dynamic Algebra}

Fix a  vocabulary of atomic module symbols $\tau_{ M}$.
Let  $\tau_{P}$, where  $\tau_{P} \subseteq \tau_{ M}$, be a vocabulary of atomic module symbols   where inputs are {\em not} specified.

We call them {\em propositions}. Alternatively, we can think of these modules as having outputs that are identical to the inputs.
Let $\tau_{act}$, where  $\tau_{act} \subseteq \tau_{ M}$, be a vocabulary of atomic module symbols $M_i(\underline{X_{i_1}}\zd X_{i_k})$,
where inputs  are underlined. We call them {\em actions}.
For one module symbol $M_i$, we can potentially have both a proposition, e.g. $M_i$,  and  several actions, depending on the choice of the inputs.
%For simplicity, we assume $vocab(M_i)\cap vocab(M_j) = \varnothing$ for $i\not = j$.

We define a calculus of binary relations  as follows. 
%\footnote{Later on, we call a two-sorted version of this calculus $L\mu\mu$, since it has 
%	two fixed points, unary and binary, to emphasize
%	a similarity with  the mu-calculus $L\mu$. }
		%	\begin{small}	
\begin{equation}
\label{eq:algebra-inf-flow}
%\begin{array}{c}
\alpha :: = \bot |    M_i?   |  M_a  |   Z_j   |  \alpha \cup \alpha  |        - \alpha |   \pi_\delta \alpha   |   \sigma_{\Theta} \alpha     |       \mu Z_j.\alpha  
%{\ \ \ \ \ \ \ \ \ \ \ \ \ \ }| \alpha   \circ  \alpha  |    \sim \alpha | 
%\end{array}
\end{equation}

%\end{small}
\noindent Here,  $M_i$ are propositions, $M_a$ are actions. Notice that the operations  are exactly like in the first algebra.

%Operation $\sim$ can be read as ``there is no outgoing $\alpha$-transitions'' and is equivalent to the standard negation 
%of a state formula   in modal temporal logics. The correspondence will be seen in the two-sorted syntax (\ref{calc-bin-rel-2-sorted}) considered below.
 Variables   $Z_j$ range over actions.
As usual, we require that $Z_j$ occurs positively (under an even number of $-$) in $ \mu Z_j.\alpha $.
Requirements on $ \pi_\delta \alpha$  and $\sigma_{\Theta} \alpha $ are as in the ``flat'' algebra.

\ignore{The calculus is an extension of  BSFP logic
\cite{DBLP:conf/aiml/ZaidGJ14} with selection and projection. However we use \emph{modules} for both unary and binary relations,
     which makes the semantics much more complicated. The non-lifted setting is a special case.
    }

%\red{??? MSO is a fragment where there is no fixed point, and quantification is limited 
%	to monadic predicates.}

\subsection{Semantics of the Dynamic Algebra}

We interpret the calculus of binary relations over 
graphs with data points that are relational databases, or graph databases, or any other data structures representable using structures.
Our data graphs are transition systems.

\begin{definition}[{Transition system} ${\cal T}$] \label{def:T} Transition system 	$${\cal T} := (V; (M_a^{\cal T})_i, (M_p^{\cal T})_j)$$
(parameterized  by valuation  $({\cal V},v)$ defined above)
 has domain $V $ which is the set of all  $\tau$-structures over a fixed domain, which is given by ${\cal V}$, and it
interprets all
actions $M_a$  as subsets of $V \times V$ denoted by $M_a^{\cal T}$,
and all monadic propositions $M_p$ by structures (now nodes in the transition graph) $M_p^{\cal T} \subseteq V$. 
\end{definition}
Module variables $Z_j$ that occur free in $\alpha$ are interpreted as actions, i.e., as subsets  of $V \times V$.

As before, we require that for any two $\tau$-structures $ \strA_1$, $ \strA_2$ which coincide on $vvoc(M_i)$,
we have  $ \strA_1 \in {\cal V}(v,M_i)$ iff  $ \strA_2 \in {\cal V}(v,M_i)$.

 %One can understand  this definition as follows. 
Recall how valuations $({\cal V},v)$ work. First, $v$ maps relational variables 
to symbols of the vocabulary $\tau$.
Second, $ {\cal V}$, parameterized with $v$, provides an interpretation to each atomic module, which is a set of structures as before (i.e., a concrete module).
In particular,  $ {\cal V}$ also specifies a domain. 

\ignore{

Actions are the transitions in the transition system $\cal T$, and propositions are labels
of the states of $\cal T$.  Sequential composition $\circ$ and non-deterministic 
choice $+$ act as expected,  projection adds non-determinism, similarly to $+$, and selection 
restricts the action to that where the interpretations of $L_1$ and $L_2$ are equal (there are three cases, when both $L_1$ and $L_2$ are 
inputs, both are outputs, and one is input, one is output), with interpretations as we would expect. 
The semantics of   $ \mu Z_j.\alpha $ is exactly like that of the least fixed point operator
in the modal mu-calculus $L\mu$. 
}

We define the extension $\llbracket \alpha \rrbracket ^{{\cal T},{\cal V},v} $  of  formula $\alpha$ in $\cal T$ under valuation $({\cal V},v)$ inductively as follows. 

\newcommand{\PLH}{{\mkern-2mu\times\mkern-2mu}}
%\vspace{-2mm}
$$
\hspace{-5mm}
\begin{array}{l}

%\llbracket \top \rrbracket ^{{\cal T},{\cal V},v} : = \{ (  \strB ,  \strC) \in V^{\cal T} \, \! \PLH\, V^{\cal T} \   \mbox{for all $ \strB$, $ \strC$ in $V$}  \}.\\
\llbracket \bot \rrbracket ^{{\cal T},{\cal V},v} : = \varnothing.\\
\llbracket M_i?\rrbracket ^{{\cal T},{\cal V},v} : = \{ (  \strB ,  \strB) \in V^{\cal T} \! \PLH  V^{\cal T} \  \ |\   \strB \in {\cal V} (v, M_i) \  \}.\\

%%%%%%%%%%%%%%%%%%%%%%%%%%%%%%%%%%%%%%%%%%%%%%%%%%%%%%%%%%%%%%%%

 \llbracket M_a\rrbracket ^{{\cal T},{\cal V},v} : =  \{ (  \strB_1 ,  \strB_2) \in V^{\cal T} \! \PLH V^{\cal T}  \ |\  \  \strB_1|_{\tau\setminus {\varepsilon_{M_a}}} =  \strB_2|_{\tau\setminus \varepsilon_{M_a}} \\
   \mbox{ and }  
    \strB_2 \in {\cal V} (v, M_a) \   \}.\\

    %%%%%%%%%%%%%%%%%%%%%%%%%%%%%%%%%%%%%%%%%%%%%%%%%%%%%%%%%%%%%%%%

\llbracket \alpha_1 \cup  \alpha_2\rrbracket ^{{\cal T},{\cal V},v} : = \llbracket \alpha_1\rrbracket ^{{\cal T},{\cal V},v} \cup  \llbracket \alpha_2\rrbracket ^{{\cal T},{\cal V},v}.\\

%%%%%%%%%%%%%%%%%%%%%%%%%%%%%%%%%%%%%%%%%%%%%%%%%%%%%%%%%%%%%%%%

    \llbracket -  \alpha_2\rrbracket ^{{\cal T},{\cal V},v} : =  V^{\cal T} \! \PLH V^{\cal T}  \setminus \llbracket \alpha\rrbracket ^{{\cal T},{\cal V},v} .\\

    %%%%%%%%%%%%%%%%%%%%%%%%%%%%%%%%%%%%%%%%%%%%%%%%%%%%%%%%%%%%%%%%

\llbracket \mu Z_j.\alpha \rrbracket ^{{\cal T},{\cal V},v}  : = \bigcap \big\{ R \subseteq  V^{\cal T} \! \PLH V^{\cal T} \ : \ \llbracket \alpha\rrbracket  ^{{\cal T},{\cal V}[Z:=R],v }\subseteq R \big\} .\\

%%%%%%%%%%%%%%%%%%%%%%%%%%%%%%%%%%%%%%%%%%%%%%%%%%%%%%%%%%%%%%%%
\llbracket \pi_{\delta}(\alpha) \rrbracket ^{{\cal T},{\cal V},v} : = \{ (  \strB_1 ,  \strB_2)  \in V^{\cal T} \! \PLH V^{\cal T}   \ |\ \\
\exists  \strC_1 \exists  \strC_2  \ (( \strC_1 ,  \strC_2 ) \in \llbracket \alpha \rrbracket ^{{\cal T},{\cal V},v},
\strC_1|_\delta =  \strB_1|_\delta \  \mbox{ and }  \   \strC_2|_\delta= \strB_2|_\delta  ) \}.\\
% % % % % % % % % % % % % % % % % % % % % % % % % % % % % % % % %

\end{array}	
$$

$
\llbracket \sigma_{L_1\equiv L_2} (\alpha)  \rrbracket ^{{\cal T},{\cal V},v} : = 
%{\ } \\
%\left
 \{ 
(  \strB_1 ,  \strB_2)  \in V^{\cal T}  \! \PLH V^{\cal T} \ | \ \mbox{3 cases:} \\
 %\left.   
	 $
 \begin{small}
$
 \begin{array}{l} 
\hspace{-7mm} \mbox{  1) }  \ ( \strB_1 ,  \strB_2 ) \in \llbracket \alpha \rrbracket ^{{\cal T},{\cal V},v}  \  \mbox{ and } \{L_1, L_2\}\subseteq \sigma_{\alpha} 
\  \mbox{ and }  \   \strB_1  \models L_1\equiv L_2     \\

\hspace{-7mm} \mbox{  2) }  \ ( \strB_1 ,  \strB_2 ) \in \llbracket \alpha \rrbracket ^{{\cal T},{\cal V},v}  \  \mbox{ and } \{L_1, L_2\}\subseteq \varepsilon_{\alpha}
\  \mbox{ and }  \   \strB_2  \models  L_1\equiv L_2   \\

\hspace{-7mm} \mbox{  3) }   L_1\in   \sigma_{\alpha} \mbox{ and }  L_2\in   \varepsilon_{\alpha}  \  \mbox{ and } \\

 \hspace{-5mm} \exists  \strC \  (\ ( \strC ,  \strB_2 ) \in \llbracket \alpha \rrbracket ^{{\cal T},{\cal V},v}  \  \mbox{ and } \  \  
\  \strB_1|_{\tau\setminus \{L_1\}} =   \strC|_{\tau\setminus \{L_1\} }  \  \mbox{ and } \ \\
( \strB_1 \models  L_1   \  \mbox{ iff }\   \strB_2 \models  L_2)\  ).  
% old condition, replased by the one above:  \strB_2 \models  (L_1\equiv L_2) ). 
\end{array}  %\right \}. 
$
\end{small}

Here, $\models$ is the standard satisfaction relation as in the first-order logic.
\noindent Case 3 expresses  feedback from  output $L_2$ to  input $L_1$, similar to a feedback in boolean circuits, also used in \cite{TT:FROCOS:2011-long}.
 Notice that  $L_1$ is a new guessed symbol, so the number of models
in the third case may increase. Cases 1 and 2 potentially reduce the number of models.

\subsubsection{Satisfaction Relation for the Calculus of Binary Relations}
\begin{definition}[Satisfaction, Dynamic algebra] \label{def:satisfaction-binary}
	Given a well-formed algebraic expression $\alpha$ defined by (\ref{eq:algebra-inf-flow}), we say that transition system  $\cal T$ and pair of states 
	$( \strA,  \strB)$ {\em satisfy} $\alpha$ under 
	valuation $({\cal V},v)$, notation 
	$${\cal T}, ( \strA,  \strB) \models_{({\cal V},v)}  \alpha,$$
	if $( \strA,  \strB) \in \llbracket \alpha \rrbracket ^{{\cal T},{\cal V},v}$. 
\end{definition}

\subsubsection{Atomic Modules-Actions}\label{subsec:NonDetAction}
Here, we clarify the semantics of the atomic modules-actions. According to the semantics, atomic actions produce a replica of the current database except the interpretation 
	of the expansion (output) vocabulary changes as specified by the action. This is similar to the inertia law in the Situation Calculus and other formalisms  for reasoning about actions.
\begin{example}{\rm 
		To illustrate transitions using our examples, in (\ref{HC-2Col-MX}), first $M_{\rm HC}( \underline{V},\underline{X},Y) $ makes transition by producing 
		possibly several Hamiltonian Circuits. The interpretation of the output vocabulary, $\{Y\} $ changes,
		everything else is transferred by inertia.   Then each  resulting structure is taken as an input to 2-Colouring, $	M_{\rm 2Col}(\underline{V},Y,Z,T)$,
		where edges in the cycle, $Y $,  are ``fed'' to the second argument of  $	M_{\rm 2Col}$, although this is hidden from the outside observer by the existential quantifier in (\ref{HC-2Col-MX}).
		The second module  produces non-deterministic transitions, one for each  generated colouring.
	}
\end{example}

\subsection{Useful Operations}
We introduce some definable operations.

\subsubsection{ Basic set-theoretic operations and equivalence}
%These are the same as in the ``flat''case.
$$
\begin{array}{l}
\top := - \bot,\\
\alpha_1 \cap \alpha_2   : = - (- \alpha_1 \cup - \alpha_2) ,\\
\alpha_1 -\alpha_2 : = - (- \alpha_1 \cup \alpha_2), \\
\alpha_1 \equiv \alpha_2 : =  (- \alpha_1 \cup  \alpha_2) \cap (- \alpha_2 \cup  \alpha_1).
\end{array}
$$
%The operations do not require  additional clarification.
%which give us, respectively:
%$$
%\begin{array}{l}
% \llbracket \alpha_1 -\alpha_2\rrbracket ^{{\cal T},{\cal V},v}   =  \llbracket \alpha_1\rrbracket ^{{\cal T},{\cal V},v} \setminus  \llbracket \alpha_2\rrbracket ^{{\cal T},{\cal V},v}.\\
%\end{array}	
%$$

\subsubsection{ Projection onto the inputs}
$$\downarrow\alpha\  : = \pi_{\sigma_{\alpha}} \alpha. $$

This operation is also called ``projection onto the first element of the binary relation''.
It  identifies the states in $V$ where there is an outgoing $\alpha$-transition.
Thus,

$$
\llbracket \downarrow\alpha \rrbracket ^{{\cal T},{\cal V},v}  = \{ (  \strB ,  \strB) \in V^{\cal T} \! \PLH V^{\cal T} | \exists  \strB'\ \  (  \strB ,  \strB')  \in \llbracket \alpha\rrbracket ^{{\cal T},{\cal V},v}   \}.
$$

\subsubsection{ Projection onto the outputs}

$$\uparrow\alpha\  : = \pi_{\varepsilon_{\alpha}} \alpha. $$

\subsubsection{  Unary negation} 
$$\sim \alpha\  : = \downarrow (-\alpha).$$
 It says ``there is no outgoing $\alpha$-transition''. By this definition,

$$
\llbracket \sim \alpha \rrbracket ^{{\cal T},{\cal V},v}  = \{ (  \strB ,  \strB) \in V^{\cal T} \! \PLH V^{\cal T} \  \ |\    \ \forall  \strB'\ \  (  \strB ,  \strB') \not \in \llbracket \alpha\rrbracket ^{{\cal T},{\cal V},v}  \}.\\
$$

By these two definitions, $\downarrow\alpha\   = \ \sim \sim \alpha$.
%\footnote{Equality ($=$) here is a meta-logic notation that denotes that the two expressions are the same semantically. }
\subsubsection{  Diagonal} 
$$D:= \sim \bot.$$

$$
\llbracket D  \rrbracket ^{{\cal T},{\cal V},v}  = \{ (  \strB,  \strB) \in V^{\cal T} \! \PLH V^{\cal T}\}.\\
$$

% $$
%	M_1? \circ 	M_2?  \ = \ M_1? \times 	M_2?,\\
%  \sim \alpha   \  =   \   (D -  \downarrow (\alpha)).\\
% $$

This operation is sometimes called the ``$nil$'' action, or it can be seen as an empty word which is denoted $\varepsilon$ in the formal language theory.

\subsubsection{Sequential composition}
It is very common, in modal logics of programs (e.g. Dynamic Logic), in expressive Description Logics, in graph databases, etc., to consider the composition 
operator, but not intersection. 
Sequential composition ($\alpha_1 \circ  \alpha_2$)  is not definable using the other operations, but is a particular case of intersection ($\cap$), and it can be obtained by imposing 
a simple sufficient syntactic restriction on the expressions combined. 
$$
\begin{array}{l}
\llbracket \alpha_1 \circ  \alpha_2 \rrbracket ^{{\cal T},{\cal V},v} : = \{ (  \strA ,  \strB)  \in V^{\cal T} \! \PLH V^{\cal T} \  \ |\ \\
\exists  \strC (( \strA ,  \strC ) \in \llbracket  \alpha_1\rrbracket ^{{\cal T},{\cal V},v} 
\mbox{ and }   (  \strC ,  \strB ) \in \llbracket \alpha_2\rrbracket ^{{\cal T},{\cal V},v}  )  \}.\\
\end{array}
$$

If   there are neither  output interference  nor 
cyclic dependencies, then intersection becomes sequential composition:
\begin{proposition}\label{prop:composition}
If $\varepsilon_{\alpha_1} \cap \sigma_{\alpha_2} \not = \emptyset$, $\varepsilon_{\alpha_2} \cap \sigma_{\alpha_1} \not = \emptyset$ and
	  $\varepsilon_{\alpha_1} \not = \varepsilon_{\alpha_2}$, then
	$$ \alpha_1 \cap \alpha_2 \ = \ \alpha_1 \circ \alpha_2.$$
\end{proposition}

\subsubsection{Counting}
This operation comes from graph databases. It represents a
path that is composed of $k$ pieces $\alpha$, where $n \leq k \leq m$ and $n<m$.
	$$
	 \alpha ^{n,m}: = \underbrace{\alpha \circ \dots  \circ  \alpha}_n \circ \underbrace{(\alpha \cup D) \circ \dots  \circ (\alpha \cup D)}_{m-n},
	$$
	where we have a composition of $n$ times $\alpha$ and $m-n$ times $\alpha \cup D$. This definition produces:
$$
\llbracket \alpha^{n,m}  \rrbracket ^{{\cal T},{\cal V},v}  = \bigcup_{k=n}^m(\llbracket \alpha\rrbracket ^{{\cal T},{\cal V},v}     )^k.\\
$$

\subsubsection{Reverse} 
This operation is common in e.g. Description Logics as well as in graph databases. 
It amounts to changing the direction of information propagation, i.e., flipping inputs and outputs. 
It is not definable in the syntax as presented here,
but notice that the operations of assigning inputs $\sigma_{\alpha}$ and outputs $\varepsilon_{\alpha}$ are silently present 
in the language (we added them to the ``flat'' algebra). They could be made explicit, and that would give us the  reverse operator.

\subsubsection{Subexpression Tests} 
These operations check if a path in the transition graph starts and ends with the same or different data values. 
$$
\begin{array}{l}
\alpha_{ =} \  : =     \ \   \downarrow\alpha \  \equiv \ \uparrow\alpha ,\\
\alpha_{\not =} \  : =     \ \ -  (\downarrow\alpha \  \equiv \ \uparrow\alpha).
\end{array}
$$
By these definitions, 
$$
\begin{array}{r}
\llbracket \alpha_{ =}  \rrbracket ^{{\cal T},{\cal V},v}  =  \{ (  \strB_1 ,  \strB_2)  \in V^{\cal T}  \! \PLH V^{\cal T} \ |   ( \strB_1 ,  \strB_2 ) \in \llbracket \alpha \rrbracket ^{{\cal T},{\cal V},v}  \ \\
 \mbox{ and }   \strB_1 =  \strB_2  \},\\

\llbracket \alpha_{\not =}  \rrbracket ^{{\cal T},{\cal V},v}  =  \{ (  \strB_1 ,  \strB_2)  \in V^{\cal T}  \! \PLH V^{\cal T} \ |   ( \strB_1 ,  \strB_2 ) \in \llbracket \alpha \rrbracket ^{{\cal T},{\cal V},v}  \  \\
\mbox{ and }   \strB_1 \not =  \strB_2  \}.

\end{array}
$$

\subsubsection{Constant Tests}

For a (constant) relation $R$ on the domain elements, we can check if $R$ is (or is not) the interpretation of a particular relational variable 
(under a variable assignment $({\cal V},v)$) using the selection operation.
$$
\begin{array}{l}
R^{ =} \  : =     \ \ \sigma_{X=`R'} D  ,\\
R^{\not =} \  : =     \ \ -  (\sigma_{X=`R'} D ).
\end{array}
$$
By these definitions, 
$$
\begin{array}{l}
\llbracket R^{ =} \rrbracket ^{{\cal T},{\cal V},v}  =
 \{ (  \strB ,  \strB) \in V^{\cal T} \! \PLH V^{\cal T} \  \ |\   \strB \models_{({\cal V},v)} X \equiv  `R' \},\\
 \llbracket R^{\not =} \rrbracket ^{{\cal T},{\cal V},v}  =
 \{ (  \strB ,  \strB) \in V^{\cal T} \! \PLH V^{\cal T} \  \ |\   \strB \not \models_{({\cal V},v)} X \equiv  `R' \}.
 \end{array}
$$

\ignore{

\subsubsection{Renaming} 

We can rename $X$ to $Q$ in $\alpha$:\\
	$$\alpha[Q/X]:= \pi_{vvoc(\alpha) \setminus\{X\}\cup\{Q\} }\sigma_{X\equiv  Q} (\alpha \times \top)$$.

\subsubsection{Extending the vocabulary of $\alpha$}
	
We can extend the vocabulary 	 of $E$ to  a bigger vocabulary:  
$$\pi_{\delta}(E\times \top).$$
	
}

\subsection{Two-Sorted Syntax, $L\mu\mu$}
The grammar  (\ref{eq:algebra-inf-flow}) for the algebra with information flow can be equivalently represented 
in a ``two-sorted'' syntax, denoted $L\mu \mu$, where expressions for state formulae $\phi$ and processes $\alpha$ are defined by mutual recursion.	
\begin{small}	
	\begin{equation}
	\label{calc-bin-rel-2-sorted}
	\begin{array}{c}	
	\hspace{-4mm}	\alpha :: =  \bot \  |  \ M_a \ |  \ Z_j  \ | \ \alpha \cup \alpha  \ | - \alpha \ | \     { \pi_\delta(\alpha)} \  | \   {\sigma_{\Theta} (\alpha) } \ |    \  \phi?   \ | \    \mu Z_j.\alpha \\
%	{\ \ \ \ \ \ \ \ \ \ \ \ \ \ }| \alpha   \circ  \alpha  |    \sim \alpha ???? | \\
	\hspace{-4mm}	\phi :: = M_i \ |  \  X_i  \ |  \ \ \phi \lor \phi \ | \  \neg  \phi \ | \ \langle \alpha \rangle \phi        \ | \  \mu X_j.\phi
	\end{array}
	\end{equation}
\end{small}
Notice that the second line corresponds to the mu-calculus $L\mu$.
We define the necessity modality through the possibility modality: $[ \alpha ] \phi: = \neg 	\langle \alpha \rangle \neg \phi   $.
Thus, we can write 	$	\langle \alpha \rangle \phi $ (respectively, $[ \alpha ] \phi$)	to express that after some (respectively, all) executions of modular system $\alpha$, property $\phi$ holds. As usual, $\phi_1 \land \phi_2 := \neg \phi_1 \lor \neg \phi_2$.	Notice that we have binary (for processes) and unary (for state formulae) fixed points.
%\hspace{-20mm}				

%\vspace{-1mm}

The formulae in this logic allow one to specify the goals of the execution, both eventual 
and extended in time.  Thus, $L\mu\mu$ can act as a programming language. 

\subsection{Semantics of $L\mu\mu$}
The  modal logic $L\mu\mu$  (\ref{calc-bin-rel-2-sorted})   is interpreted over the same  transition system as the calculus of binary relations (\ref{eq:algebra-inf-flow}).

{\bf State Formulae:} Exactly like in the $\mu$-calculus:
\begin{small}
	$$
	\hspace{-5mm}
	\begin{array}{l}
	\llbracket M_i\rrbracket ^{{\cal T},{\cal V},v} : = \{ \strA \in V^{\cal T}  \  \ |\  \ \strA \in {\cal V} (v, M_i) \  \},\\
	\llbracket \phi_1 \lor  \phi_2\rrbracket ^{{\cal T},{\cal V},v} : = \llbracket \phi_1\rrbracket ^{{\cal T},{\cal V},v} \cup  \llbracket \phi_2\rrbracket ^{{\cal T},{\cal V},v},\\
	
	\llbracket \neg \phi\rrbracket ^{{\cal T},{\cal V},v} : =  \{ \strA \in V^{\cal T}  \  \ |\  \ \strA \not \in \llbracket  \phi\rrbracket ^{{\cal T},{\cal V},v}\  \},\\
	\llbracket \langle \alpha \rangle \phi\rrbracket ^{{\cal T},{\cal V},v} : =  \{ \strA \in V^{\cal T}  \  \ |\  \ \exists \strB \ ( ( \strA , \strB)  \in V^{\cal T} \! \PLH V^{\cal T} \ \mbox{ and } \\
	\ ( \strA , \strB)  \in \llbracket  \alpha \rrbracket ^{{\cal T},{\cal V},v} \mbox{ and } \strB  \in \llbracket  \phi\rrbracket ^{{\cal T},{\cal V},v} ) \  \},\\
	\llbracket \mu Z_j.\phi \rrbracket ^{{\cal T},{\cal V},v}  : = \bigcap \big\{ R \subseteq  V^{\cal T} \! \ : \ \llbracket \phi \rrbracket  ^{{\cal T},{\cal V}[Z:=R],v }\subseteq R \big\} .\\
	\end{array}
	$$
	
\end{small}

{\bf Process Formulae:} Exactly like in the one-sorted syntax, and, in addition, like in Dynamic Logic: 
%\begin{small}
	$$
	\hspace{-5mm}
	\begin{array}{l}
	\llbracket \phi ?\rrbracket  ^{{\cal T},{\cal V},v} : =  \{ (\strA,\strA ) \in V^{\cal T} \! \PLH V^{\cal T}  \  \ |\  \ \strA  \in \llbracket  \phi\rrbracket ^{{\cal T},{\cal V},v}\  \}.
	\end{array}
	$$
%\end{small}

\subsubsection{Example: Equality Test}
\begin{example}[Equality Test]
{\rm
	Formula 
	$ \langle\alpha_1\equiv \alpha_2 \rangle  \top$  specifies the set of states from where some executions of  $\alpha_1$ and  $\alpha_2$ 
	lead to the same data value. Here, $\top$ is an abbreviation for any tautology, e.g. $M\lor \neg M$.
By the semantics, the meaning  of this formula is:
$$
\begin{array}{l}
\llbracket \langle\alpha_1\equiv \alpha_2 \rangle  \top \rrbracket ^{{\cal T},{\cal V},v}  = \{ (  \strB ,  \strB) \in V^{\cal T} \! \PLH V^{\cal T} \  \ |\  \\
\ \exists  \strB' \exists  \strB''\ \ (\  (  \strB ,  \strB')  \in \llbracket \alpha_1\rrbracket ^{{\cal T},{\cal V},v} \ \  \mbox{and} \  \ 
(  \strB ,  \strB'')  \in \llbracket \alpha_2\rrbracket ^{{\cal T},{\cal V},v} \\
{ \ \ \ \ \ \ \ \ \ \ \ \ \ \ \ \ \ \ \ \ \ \ \ \ \ \ \ \ \ \ \ \ }\ \mbox{and} \ \strB =  \strB')
\}.\\
\end{array}
$$	
%When both  $\alpha_1$ and  $\alpha_2$  are deterministic, 
This operation corresponds to an operation used in graph databases, XPath. Non-equality test has a negation ($-$) in front of the equivalence  ($\equiv$).	
}
\end{example}

\subsubsection{Satisfaction Relation for  $L\mu \mu$}
\begin{definition}[Satisfaction Relation, $L\mu \mu$] {\ } \\
	Given a well-formed  state formula $\phi$ and process formula $\alpha$ as defined by (\ref{calc-bin-rel-2-sorted}), we say that transition system  $\cal T$ and  state 
	$\strA$ {\em satisfy} $\phi$ under 
	valuation $({\cal V},v)$, notation 
	$${\cal T}, \strA \models_{({\cal V},v)}  \phi,$$ 
	if $\strA \in \llbracket \phi \rrbracket ^{{\cal T},{\cal V},v}$.
For process formulae $\alpha$, the definition  of the satisfaction relation is exactly as in Definition \ref{def:satisfaction-binary}. 
	
\ignore{	Similarly, we say that transition system  $\cal T$ and a pair of states 
	$(\strA, \strB)$ {\em satisfy} $\alpha$ under 
	valuation $({\cal V},v)$, notation 
	$${\cal T}, (\strA, \strB) \models_{({\cal V},v)}  \alpha,$$
	 if $(\strA, \strB) \in \llbracket \alpha \rrbracket ^{{\cal T},{\cal V},v}$.	
}
	
\end{definition}

%\subsubsection{The Two Presentations of the Syntax are Equivalent}

\subsubsection{ Two-Sorted = Minimal Syntax  }
The two representations of the algebra (one-sorted and two-sorted) are equivalent, as we show below.
The statement is analogous to the one in  
\cite{DBLP:conf/aiml/ZaidGJ14}.
\begin{theorem}
	For every state formula $\phi$ in two-sorted syntax (\ref{calc-bin-rel-2-sorted}) there is a formula $\hat{\phi}$ in the minimal syntax (\ref{eq:algebra-inf-flow})
	such that ${\cal T}, \strB \models_{({\cal V},v)}  \phi $ iff ${\cal T}, (\strB, \strB) \models_{({\cal V},v)}   \downarrow \hat{\phi}$,
	and for every action formula $\alpha$ there is an equivalent formula $\hat{\alpha}$ in the minimal syntax.
\end{theorem}	

\begin{proof}
We need to translate all the state formulae into process formulae.	We do it by induction on the structure of the formula.
Atomic constant modules and module variables remain unchanged by the transformation. However, monadic variables are considered binary. 
	\begin{compactitem}
		\item If $\phi = \phi_1 \lor \phi_2 $, we set $ \hat{\phi}:= \hat{\phi_1} \cup \hat{\phi_2}$.
		\item If $\phi = \neg \phi_1 $, we set $ \hat{\phi}:= \sim \hat{\phi}_1$.
			\item If $\phi = \langle \alpha_1 \rangle \phi_1 $, we set $ \hat{\phi}:= \ \hat{\alpha_1} \circ \hat{\phi_1}$.
			\item If $\phi = \mu X. \phi_1 $, we set $ \hat{\phi}:= \mu X.\downarrow \hat{\phi_1}$.
		
	\end{compactitem}
All process formulae $\alpha$ except test $\phi_1?$ remain unchanged under this transformation. For test, we have:
	\begin{compactitem}
	\item If $\alpha = \phi_1?$, we set $\hat{\alpha} : = \downarrow\hat{\phi_1}$.
\end{compactitem}
It is easy to see that, under this transformation, the semantic correspondence holds.
\end{proof}

\ignore{

\subsection{Sources of Nondeterminism}

We have two sources of nondeterminism: nondeterminism that comes from atomic modules-actions,  and nondeterminism that comes
from the algebraic operations. Added nondeterminism can increase expressive power, and therefore, computational complexity, so it is important to 
understand its sources.

\subsubsection{Power-preserving operations}
Languages with polytime data complexity are  considered very attractive for data analysis. 

When all modules are deterministic and polytime computable, then, with power preserving operations, 
we keep polytime data complexity. 

Power-preserving operations are sequential composition, projection onto the output vocabulary, 

\subsubsection{Power-increasing operations}

}

\subsection{Some Notable Fragments}
Since our logic is very expressive, it is not surprising that many logic are fragments of it.
For example, the well-known temporal logics CTL, LTL, CTL$^*$ are obviously  
fragments of $L\mu\mu$ since they are fragments of the mu-calculus $L\mu$. 
What is interesting, however, is to analyze examples of diverse nature and origin where efficient reasoning is the goal.
Often, information propagation plays a  role there, and a modal logic is obtained as a result. 
We already saw examples of operations used in graph databases. Graphs are, obviously, binary relations. 
We now consider several other well-known logics.
\subsubsection{Dynamic and Description Logics}
Dynamic Logic was created for reasoning about programs.
\begin{proposition}
	The  Dynamic Logic 
	\begin{small}	
		\begin{equation}
		\label{DL}
		\begin{array}{l}	
		\hspace{-4mm}	\alpha :: =  M_a   \ | \ \alpha \cup \alpha \ | \  \alpha   \circ  \alpha \ | \      \  \phi?   \ | \    \alpha^* \\
		\hspace{-4mm}	\phi :: = M_i \ |    \ \ \phi \lor \phi \ | \  \neg  \phi \ | \ \langle \alpha \rangle \phi       
		\end{array}
		\end{equation}
	\end{small}is a fragment of (\ref{calc-bin-rel-2-sorted}). 
\end{proposition}
\begin{proof}
	%It  follows from \cite{DBLP:conf/aiml/ZaidGJ14} because  DynLFP ( $L\mu\mu $?????) is an extension of BSFP introduced there.
By Proposition \ref{prop:composition}, sequential composition is a particular case of intersection.	
The other operations of Dynamic Logic  are a subset of those in  (\ref{calc-bin-rel-2-sorted}).
Thus, it is sufficient to express $ \alpha^*$.  We have $ \alpha^*:= \mu Z. (D \cup Z \circ \alpha)$.
\end{proof}

The following property follows from the well-known connection between expressive Description Logic and Dynamic Logic with reverse operator  
 \cite{DescriptionLogic-Calvanese}. % and other chapters and references in the same volume). The connection was first published by Schild in 1991.
%The connection was first observed by ..... 

\begin{corollary}
	 Description Logic ${\cal ALCI}_{reg}$
		\begin{small}	
			\begin{equation}
			\label{DescrL}
			\begin{array}{l}	
			\hspace{-4mm}	R :: =  P   \ | \ R \cup R \ | \  R   \circ  R \ | \      \  id(C)   \ | \    R^*  \ | \  R^{-}  \\
			\hspace{-4mm}	C :: = A \ |    \ \ C \lor C \ | \  \neg  C \ | \ \exists R. C       
			\end{array}
			\end{equation}
		\end{small}is a fragment of (\ref{calc-bin-rel-2-sorted}). 
\end{corollary}
In (\ref{DescrL}), $C$ denotes concepts, $R$ denotes roles, and $A$ and $R$ stand for atomic concepts and roles, respectively.
Notation $ id(C)$ stands for test, $R^{-}$ is reverse operator, $ \exists R. C $ is the modal ``exists'' operator.

\subsubsection{From Second-Order to First-Order}\label{subsubsec:SO-FO}
Here, we explain how atomic modules become standard predicates in the sense of first-order logic.
Propositional logic is a fragment of first-order logic. To see this,  view propositions as predicate symbols 
over one-element domain. Similarly, first-order logic is a fragment of second-order logic. In second-order logic, 
we allow quantification over relations, i.e., over sets of tuples of domain elements. If every such set is a singleton,
and every relation we quantify over is unary, then second-order quantifiers behave as quantifiers over domain elements,
and second-order collapses to first-order. Our modules are sets of structures.
When second-order logic collapses to first-order, {\bf each structure becomes a tuple of domain elements},
and each module becomes a relation in the sense of first-order logic. In this case, also, third-order fixed points
that represent sets of structures, collapse to the usual fixed points that are relations (sets of tuples).
Second-order logic is a great tool that allows us to talk about many things uniformly.

\subsubsection{ Datalog$^+_-$ as a Modal Logic}

Just as  our calculus, Datalog$^+_-$  \cite{Datalog+/-} forces multi-dimensional expressions into binary. This is not immediately apparent,
unless one carefully examines the rules  in search of explicit or implicit existential quantifiers. 
Those quantifiers produce possibility 
modalities, which turn into necessity under negation.  For example, the following program
$$
\begin{array}{l}
emp(X) \rightarrow \exists Y hasMgr(X,Y), emp (Y)\\
 {\it person}(P) \rightarrow \exists F\  {\it fatherOf} (F,P)\\
 {\it fatherOf} (F,P) \rightarrow  {\it person}(F)
  \end{array}
$$
translates into the following formula in our calculus:
$$
\begin{array}{l}
(emp(X) \rightarrow \langle hasMgr(\underline{X},Y)\rangle \ emp (Y))\\
\land ({\it person}(P) \rightarrow \langle  {\it fatherOf} (F,\underline{P})\rangle \ \top)\\
\land (\lbrack {\it fatherOf} (F,\underline{P}) \rbrack \  {\it person}(F)),
 \end{array}
$$
where atomic module symbols are elements of the set  $\delta:= \{ {\it emp, hasMgr, person, fatherOf }     \}$, and implication ($\rightarrow$) is the usual abbreviation. 

The semantics of  Datalog$^+_-$ is given by the least model that satisfies the rules (which are in the Horn form).
By a well-known construction, this least model is expressible by a simultaneous fixed point $\mu \delta . \phi$.
The most common reasoning task for Datalog$^+_-$  is {\em certain} reasoning, which amounts to computing if a query is true in every expansion of a given database that 
satisfies a given Datalog$^+_-$ program. 

The language of Datalog$^+_-$ can be greatly enriched following $L\mu\mu$ (\ref{calc-bin-rel-2-sorted}), by allowing compound expressions inside the modalities, 
e.g. regular expressions or the special-purpose operations we gave as examples. 
\ignore{Embedding  Datalog programs into classical logic is also a good extension.
Embedding fixed point constructs, which are inherent to the Datalog family, into classical logic is also a good extension as purely fixed-point based logics are not always compositional.\footnote{For example, lack of compositionality is a major issue for Answer Set Programming under stable model semantics.}
}

%$\llbracket$ and $\rrbracket$ from the stmaryrd package
%$\textlbrackdbl$ and $\textrbrackdbl$ from the textcomp package

%\input{Files/3-Colouring-Example}

%\thickmuskip=0.5\thickmuskip

\section{ Queries, Machines, Modalities}\label{sec:Queries-Machines-Modalities}

In this section, we connect two declarative ways of specifying problems, as in database query answering and as in temporal
logic model checking, with each other and with a machine-based approach that has an imperative flavour. The connection between the three approaches  is possible because of information propagation.

%we introduce simple computing devices. The modal logic from the previous section acts as a programming language for these machines.
%We also analyze factors affecting computational complexity.
\subsection{Structures as Computing Devices}

\subsubsection{Abstract Machines} 

% For deeper complexity analysis, 
We now introduce very {\bf simple  {abstract machines}}, similar to automata. %,  that represent declarative problem solving. 
The main (and only) operation of these machines is  a task broadly solved in practice, the Model Expansion task \cite{MT05-long}, see Definition \ref{def:MX}. 
Our computing devices are {\em $\tau$-structures}. All they do is {\bf store information and expand}.  
Declarative specifications of modular systems (formulae)
are the programs for these simple devices, and the modal logic from the previous section can be viewed 
as a programming language for these machines. 
%We call such computations {\em natural} as they do not appeal to structure encodings on a tape of a Turing machine.  
Program {\em execution}  consists of {\em constructing a transition system}.

\subsubsection{Transition System Determined  by $\alpha$}
Without loss of generality, assume that all atomic modules are represented by binary relations, as in the semantics of the calculus of binary relations (\ref{eq:algebra-inf-flow}). To talk about all possible executions of $\alpha$, we construct a transition system \TS{\ }by starting from 
 ${\cal T} := (V; (M_a^{\cal T})_i, (M_p^{\cal T})_j)$,   as in Definition \ref{def:T}, and adding labelled edges produced by each subformula.
\begin{definition}[Transition System \TS]
	Given formula $\alpha$ in the calculus of binary relations and valuation $({\cal V},v)$, where $v$ maps $vvoc(\alpha)$ to $\tau$, the {\em labelled transition system  \TS {\ }that represents 
		possible executions of $\alpha$} is
	$$
	\mbox{ \TS }:= (V; L),
	$$
where $V$ is the set of $\tau$-structures,  
and  $L$  is a set of labelled edges.  The edges are constructed according to the following rule:
If    $\alpha_i$ is  a subformula of $\alpha$, then 

$$	
 (  \strB_1 ,  \strB_2)  \in L \ \ \ \  \mbox {iff}  \ \ \ \    (  \strB_1 ,  \strB_2) \in  \llbracket \alpha_{i}  \rrbracket ^{{\cal T},{\cal V},v}, 
$$
and the label of 	$ (  \strB_1 ,  \strB_2)  $ is $\alpha_i$.
\end{definition}

Notice that since valuation  $({\cal V},v)$ is given, function ${\cal V}$ specifies the domain and interpretations of the atomic symbols.
Thus, generating \TS{\ }is a constructive process. 
%However, if  $({\cal V},v)$ was not given, such a construction would not be possible. 

\subsubsection{Complexity Measures}\label{subsubsec:complexity-measures}
The time and space complexity of constructing \TS{\ }is associated with the complexity of 
satisfying $\alpha$ over a given domain (which is the Model Expansion task).
While data, expression and combined complexity are considered the main measures of the amount of computation required,
we argue that output complexity (in the sense of output-sensitive algorithms)
is hugely important as well because it affects the size of the transition system
and the number of steps required.
	Recall that in the basic labelling algorithm that is in the foundation of symbolic model checking, 
	three parameters are multiplied: the number of  vertices, the number of edges, and the size of the formula. 
Data complexity is responsible for the number of vertices, expression complexity measures the size of the formula,
and output complexity is responsible 
for the number of edges. We believe that {\em input width} and {\em output width} of a formula should be considered 
separately. The former affects data complexity, the latter affects output complexity.
If we consider all  these parameters in interaction, we may be able explain, e.g., why 
some algorithms for PSPACE-complete problems (such as model checking) work reasonably well in practice, while some  algorithms 
for provably polynomial time problems behave 
very badly.

%The capturing a complexity class property,  applied e.g. to NP, shows that, for a given language: (a) {\em we can express all of NP} -- which gives an assurance of universality of the language for the given complexity class, (b) {\em no more than NP can be expressed}  -- thus solving can be achieved by means of constructing a uniform polytime reduction to an NP complete problem such as e.g. SAT. 

\subsubsection{Executions for a Given Input}

When a particular input structure $\strA$ is given,  a concrete execution materializes.
In this case, we can connect reachability in \TS {\ }with  executing $\alpha$ in the following sense.

\begin{definition}[Reachability in \TS]  {\ }\\

We say that state formula $\phi$ is reachable from the initial state $\strA$ by the execution of $\alpha$, notation	\Reachphi, if,
in the transition system constructed by executing $\alpha$, there is an edge, labelled with $\alpha$,
 from the state $\strA$ to a state where $\phi$ holds.

\end{definition}

From the construction of \TS, it follows that one has to construct the edges for all the subformulae of $\alpha$, in order to construct the edge for $\alpha$. 

We will be interested in the case where $\phi$ is a conjunction of ground atoms $ \bigwedge \bar{\cal E} $.

\ignore{
	\blue{	data, expression and combined complexity are studied under the simplifying assumption that 
	outputs are ignored. We argue that outputs should not be ignored as they determine the size 
	of the evaluation graph (and thus the complexity of the evaluation)}
}

\subsection{General Evaluation Problem \MXE } \label{def:MXE}

Recall that model checking is a special case of model expansion where the input structure interprets the entire vocabulary 
(second-order variable vocabulary in our case).
 We now define another problem that is essentially model checking -- a higher-order counterpart of the Query Evaluation Problem \QE. 
%This problem is another version of the model expansion task that is essentially model checking for ..........
%for formulae $\phi(\bX,\bY)$ with tuples of free relational variables $\bX$ and $\bY$. 
%In this version,   a mapping $v$ of relational variables   to a particular vocabulary is given on the input, or it can be fixed in advance.   ??????

\begin{definition}[{\sf Evaluation Problem \MXE}]{ \ }
	
	\underline{Given:} 
	
	\begin{compactenum}
			\item valuation $({\cal V}, v)$,
		\item formula  $\phi({\bX})$ in the calculus without information propagation,
		\item $\sigma$-structure $\strA$, where $\strA =(A; \bR_{\sigma}) $ and $\bR_{\sigma}$ interprets a part 
		of the visible (free) relational variables of $\phi$,
	
		\item a tuple of relations $\bar{\cal E}$ that interprets the rest of the visible (free) relational variables of $\phi$.                 % that match, with their arity, $vvoc(\phi)\setminus \sigma$
	%	\item that interpret a part of the remaining vocabulary?
		
	\end{compactenum}

	\underline{Find: }  Structure $\strB$ such that  
	$$\overbrace{(\underbrace{A,\bR_{\sigma}}_{\strA},\bR, \bar{\cal E})}^{\strB}\models_{({\cal V},v)} \phi( {\bX},  \bY)[\ \bR_{\sigma}/{\bX}, \   \bar{\cal E}/\bY\ ] ?$$ 
Tuple of relations $\bR$ interprets the internal (not free) relational variables of $\phi$.	We require, as usual, that for the substitutions, the arities have to coincide. Note that $\sigma$ may be empty. The domain of $\strA$ is the same as the one 
	given by $\cal V$, and $\sigma$ must be a subset of $\tau$, which is the concrete vocabulary provided by $v$ for the relational variables.
\end{definition}

Notice that {\bf this task imposes a direction of information propagation}, thus it turns an expression in a ``flat'' algebra to one in the ``dynamic'' one.
The General Evaluation Problem	(\MXE){\ }is equivalent to the Model Checking Problem (\MC){\ }for a formula where the internal relational symbols are  second-order-existential quantified. Model Expansion (\MX){\ }for $\phi({\bX})$ is equivalent to first guessing the output relations  $\bar{\cal E}$    and then using  \MXE{\ }to check.

\begin{proposition}
Query Evaluation problem	\QE{\ }is a particular case of the General Evaluation problem \MXE.
\end{proposition}

\begin{proof}
The statement	follows immediately from our explanation about why first-order logic 
is a fragment of second-order in Subsection \ref{subsubsec:SO-FO}. 	
	\end{proof}

\begin{remark}
	The definitions of the General Evaluation task \MXE{\ }and the Model Expansion task \MX{\  }could be given with partial interpretations of the predicates on the input (e.g. in terms of sets of ground atoms or 3-valued structures). That version is more convenient in other contexts, but is not necessary here. 
\end{remark}

 \subsection{Temporal Logic Tasks: \TMC, \TSAT }
 Here we introduce counterparts of Model Checking and Satisfiability tasks in the context of modal temporal logics.
 
 \begin{definition}[Model Checking: \TMC]{\ }\\
 	\underline{Given:} valuation $({\cal V}, v)$, transition system ${\cal T}$, $\sigma$-structure $\strA$, where $\sigma \subseteq \tau$,
 	state formula $\phi$. \underline{Decide:} ${\cal T}, \strA \models_{({\cal V},v)}  \phi$.
 \end{definition}
 
 A common version of this problem is the one where one is  asked to compute all the states where the formula is true. This version is 
 used in practical model checking algorithms for temporal logics. We restrict our attention to the problem of finding some structure of that sort.
  \begin{definition}[\TMC-SEARCH]
 	\underline{Given:} valuation $({\cal V}, v)$, transition system ${\cal T}$, 
 	state formula $\phi$. \underline{Find:} structure $\strA$ such that  
 	$$
 	{\cal T}, \strA \models_{({\cal V},v)}  \phi.
 	$$
 \end{definition}
  The transition system on the input is determined by the valuation and, in a way, is redundant. The valuation also gives the domain of $\strA$.
 %Notice that valuation $({\cal V}, v)$ gives the domain and determines the transition system ${\cal T}$.

 And here is another important problem that looks similar on the surface, but can be drastically different computationally.
 \begin{definition}[Satisfiability: \TSAT] {\ }\\
 	\underline{Given:}  State formula $\phi$, valuation $v$ that fixes a concrete vocabulary. 
 	\underline{Find:} valuation ${\cal V}$, transition system ${\cal T}$, structure $\strA$  (one of the states of  ${\cal T}$), such that 
 	$${\cal T}, \strA \models_{({\cal V},v)}  \phi.$$
 \end{definition}
 
The main difference here is that we need to find a valuation ${\cal V}$, which determines the domain and the interpretation of ``unary'' and ``binary''
module symbols. The transition system  ${\cal T}$  is then constructed from those.
 
 \subsection{Connecting Machines, Queries, Modalities}
 
 In this subsection, we show that, in our translation from ``flat'' to modal logic,  the \TMC{\ }task is the same as the expansion-evaluation task \MXE,
 and is also equivalent to the reachability in the execution graph.  Surprisingly, assigning input and outputs to the internal variables does not matter.
 
% ..... the transition system \TS{ \ }is determined by the formula $\alpha$, and is not on the input.  
  %	 get reachability problem. % $P(a)= \exists x (x=a \land P(x))$

 \begin{theorem}
 	Suppose we are given,  on the input, a formula  $\phi$ in the ``flat'' algebra,  $({\cal V},v)$, $\strA$ and  $\bar{{\cal E}}$.
 For any  assignments of inputs and outputs to the internal (not free) variables $\bR$ of $\phi$ that produces $\alpha$ from $\phi$, we have:
  $$
 \begin{array}{ccc}
  $\TMC$ && $\MXE$ \\ 
   \overbrace{{\cal T}, {\strA} \models \  \langle  \alpha  \rangle \bigwedge \bar{\cal E} } &    \Leftrightarrow   & \overbrace{({\strA},\bR ,\bar{\cal E} ) \models {\phi} ( \bar{\cal E}   )  } \\
  %  & \Leftrightarrow    &    {\sf ACCEPT}_{\red{\phi}}^{\sigma} ({\strA},\bar{\cal E} )  &  \Leftrightarrow   & 
    \Leftrightarrow    &  $\REACH$ &  
    %{ REACH}_{\alpha}^{\sigma} ({\strA},\bigwedge \bar{\cal E} ) & 
 \end{array}
 $$
 \end{theorem}
 
\ignore{ 
   $$
   \begin{array}{ccccc}
   TMC &&&& MXE \\ 
   \overbrace{{\cal T}, {\strA} \models \  \langle  \alpha  \rangle \bigwedge \bar{\cal E} } & \Leftrightarrow    &    {\sf ACCEPT}_{\alpha}^{\sigma} ({\strA},\bar{\cal E} )  &  \Leftrightarrow   & \overbrace{({\strA},\dots ,\bar{\cal E} ) \models \red{\phi} ( \bar{\cal E}   )  } 
   %  & \Leftrightarrow    &    {\sf ACCEPT}_{\red{\phi}}^{\sigma} ({\frak A},\bar{\cal E} )  &  \Leftrightarrow   & 
   \end{array}
   $$
}

 It follows that temporal logic model checking (symbolic, SAT-based, etc.) can be used for solving the Evaluation problem. 
  Vise versa, techniques developed for query evaluation, may influence model checking algorithms, for the corresponding fragments.

\begin{proof}
	We explain the main idea. The rest follows from the definitions of the tasks.
 Suppose a direction of information propagation for the internal variables  has been assigned. 
	The transition system{\ }\TS{\ }is constructed bottom up on the structure of the formula. We first guess the extensions
	of the input variables and identify the set 
	of states with those guessed extensions. Then we propagate the information according to each of the modules-actions,
	from inputs to the outputs. This step gives us the extensions of all the output variables of the atomic modules, and the corresponding states.
	The construction proceeds up on the structure of the formula, adding more transitions. This process is equivalent to 
	guessing all the internal variables first, then checking if the guess satisfies the atomic modules, then proceeding as before.
	Thus, the order of information propagation for the internal variables does not matter.	\end{proof}

\ignore{
\begin{corollary}  ???????????????
	The combined complexity of {\  }  \MXE {\  } is equal to the complexity of temporal modal checking {\  }  \TMC {\  } 
	for reachability queries.
\end{corollary}
}
We can think of $\phi$ as a specification of an algorithm, 
and of a corresponding $\alpha$, where the information propagation to the internal variable 
is fixed, as an implementation of it. 
The direction of information propagation is implementation-dependent.
While it is not important in theory,  it does affect the practical complexity of algorithms,
as it influences the complexity measures discussed in Subsubsection \ref{subsubsec:complexity-measures}.

A connection between database query evaluation and temporal model-checking has been noticed before.
In  footnote 10 of the paper ``From Church and Prior to PSL'' \cite{Vardi-Church-Prior-PSL},
Moshe  Vardi wrote that the two problems are ``analogous'', and that ``the study of the complexity of database query evaluation
started about the same time as that of model checking''. However, we are not aware of any papers where a precise correspondence has been established. 
This correspondence is closely related to an important question we are going to discuss next.

\ignore{

\begin{quote}
	The model-checking problem is analogous to database query evaluation, where we check the
	truth of a logical formula, representing a query, with respect to a database, viewed as a finite
	relational structure. Interestingly, the study of the complexity of database query evaluation
	started about the same time as that of model checking [122]. \cite{Vardi-queries}
\end{quote}

M.Y. Vardi. The complexity of relational query languages. In Proc. 14th ACM Symp. on
Theory of Computing, pages 137–146, 1982.

}

% % % % % % % % % % % % % % % % % % % % % % % % % % % % % % % % % % % % % % % % % % % % % % % % % % % % % % % % % % % % % %

\begin{figure*}[!htbp]
	
	\begin{center}
		\begin{tabular}{|c|c|} \hline
			{\bf Classical RA} & {\bf Lifted RA} \\ \hline \hline
			Basic units are Relations = & Basic units are  Modules =  \\ 
			sets of tuples of domain elements &  classes (sets, if the domain is fixed) of structures \\ \hline
			Object variables & Relational variables \\ \hline
			
			%		Circuits: as in nonuniform AC$^0$,  & 	Circuits: as in ``Very nonuniform'' AC$^0$ (PH),  \\ 
			%		input structures are encoded & 	input structures are ``as is'' \\ \hline
			
			{\sf Query Evaluation} task, \QE &{\sf  Evaluation} task, \MXE \\ 
			
		check if $\bar{a}$ is in the	relation  defined w.r.t. $\cA$  & check if $\bar{\cal E}$ is in an expansion of  $\cA$ \\ \hline
				{\sf Query Computation:}	mismatch	& {\sf Model Expansion, \MX:} no mismatch	\\	
							inputs are structures,  outputs are relations	&  both inputs and outputs are structures	\\	\hline
						hard to connect to modal logics & connection to modal logics is straightforward 	\\	\hline	
			Fixed points = sets of tuples,  &	Fixed points = sets of structures,  \\ 	
				$\mu Z. \phi(Z) $ is a relation &	$\mu Z. E(Z) $ is a set of structures \\ \hline 
			
					\multicolumn{2}{c}{{  }}\\	
			\multicolumn{2}{c}{{With Information Propagation}}\\	 \hline
			Resulting modal logic is ``propositional'' & Resulting modal logic is ``first-order'' \\ \hline 
			States in TS are one-element structures  & States in TS are structures \\ \hline 
				%		Actions are expansions & Actions are expansions \\ \hline
						
				  		\end{tabular}
				  	\end{center}
				  	\caption{Comparison of classical and lifted setting for Relational  Algebra (RA) with recursion}
				  	\label{fig:summary}
				  \end{figure*}

\subsection{On Robust Decidability of Modal Logics}\label{sec:ModalRobustDecidability}
Modal logic, starting from the simplest modal logic ML where the necessity and possibility modalities are added 
to the propositional logic, to much more complex  logics with path and state quantifiers and fixed points such as CTL, LTL, $L\mu$, are known for their good
computational properties.
% In \cite{Vardi-robust-decidability}, 
 Moshe Vardi posed the question to identify the main reasons for the robust decidability of modal logics,
 and  partially answered it \cite{Vardi-Modal-Robust}. In a paper with the same title \cite{Graedel-Modal-Robust},
Erich Gr\"{a}del discussed the problem further. He states the main motivation for this research as: ``We would like to have more powerful logics than ML, CTL and even the $\mu$-calculus that retain the nice properties of modal logics.'' 
Vardi also wrote that ``... modal logic, in spite of its apparent propositional syntax, is essentially a first-order logic, since the necessity and possibility modalities  quantify over the set of possible words ...'' model-checking problem for the  modal logic ML can be solved in linear time,
 while satisfiability is PSPACE-complete. On the other hand, validity, and thus, satisfiability, is highly undecidable.

Previous  partial explanations of robust decidability of modal logics are related to  two-variable fragment of first-order logic FO$^2$, 
 finite model property, tree model property, guarded fragments of first-order fragments, bisimulation-invariance and the characterization theorems.
 However, all of these explanations view modal logics as interpreted over Tarskian structures.
Then states are viewed as domain elements of a ``flat'' Tarskian structure, propositions are predicates ranging over domain elements.
Still, in all these explanations,
the complexity of satisfiability remains 
 hugely different in the modal and in the classical setting.

%in the translations from  classical to the modal logics.
%Actions are hard to match with tasks and states,  Guards (or edges) come from the ``outside'', 

% Modal logic ML, despite its apparent propositional
%syntax, is essentially a first-order logic, since the necessity and possibility modalities quantify over possible words.

While this subject deserves a separate paper, we provide a brief explanation here. 
In our understanding, possible words are structures, and the modalities quantify over those, not over domain elements.

%Treating modal logics as interpreted over Tarskian structures  is also related another research direction -- to the Characterization theorems. 
%It is an interesting orthogonal direction because 
%%it does not allow full multi-level semantic power of modal logics. It assumes the same ``flat'' level semantically,
%and asks what do we need to obtain a modal logic? 

%Our explanations:

%	Drawing an analogy with crystal formation, 
We believe that the following  is needed for a computationally good behaviour of a logic:
		\begin{itemize}
		\item information flow,
		\item an initial structure.
	\end{itemize}

%A different kind of satisfiability is used in the modal case. 

%Structures of structures, ``time crystals'',  States are structures, Propositions are relations over one-element domain,
%Actions are expansions, Guards (or edges) are subformulae, Satisfiability = Model Checking (in the Propositional Case)

%\begin{tcolorbox}
	
%\end{tcolorbox}	
The initial structure can just be a domain.
Of course, if we have a domain, the task is no longer Satisfiability, but Model Expansion. 
For the same logic, the latter task is of significantly lower complexity.
In  the propositional  case, we get the domain for free, as is explained below, even when we consider Satisfiability task.

We would now like to draw an analogy with physics. The process of producing a modal logic from a classical one  is very 
similar to the well-understood physical process of crystal formation from a liquid matter.  
The conditions needed for crystals to form requires the presence of a force (of information flow here),
and a seed (a structure in our case). During crystal formation, energy gets released.
Released energy is also responsible for a special form of  perpetual motion in the fascinating concept of time crystals invented by 
a Nobel laureate Frank Wilczek in 2012.

\ignore{
	
		\begin{itemize}
			\item A ``force'' (e.g., of information flow)
			\item A ``seed'' (an initial structure)
		\end{itemize}A transition system is a structure of structures, it is not a ``flat'' structure in the Tarskian sense. 
	In the modal case, we have more power in the semantics, we need less power in the syntax -- e.g. we drop from FO(LFP) to modal-propositional with LFP. 
	This is very similar to crystal formations -- when matter crystallizes, energy is released.}

\subsubsection{Propositional Case}

By the propositional case, we mean a modal logic built on top of propositional. 
Typical examples include CTL, LTL, $L\mu$, PDL, but we also allow all the operations of  $L\mu\mu$. 
As in the more general case, the  states in the transition system $\cal T$ are structures.

%In this case, the structure $\strA$ (and all other structures in the transition system $\cal T$)
%has a domain that consists of one element, and the propositional variables $x_1$, \dots, $x_k$ in $\phi$ are interpreted by unary relations over that domain.
%Thus, each propositional truth assignment is a structure.
%Even in the propositional case, state formulae correspond to sets of structures, process formulae correspond to sets of pairs of structures.
%Recall that valuation ${\cal V}$ interprets unary modules as sets of structures, and binary modules as pairs of structures.
%In the propositional case, there are at most $2^l$ unary and $2^l \cdot 2^l$ binary modules, where $l$ is the number of propositional variables.

%\begin{lemma}
If a propositional  $L\mu\mu$ formula is satisfiable, it is satisfiable over a transition system built from structures over one-element domain.	
%\end{lemma}
%\begin{proof}
	%Intuitively, if such a formula is satisfiable, it is satisfiable over a transition system where the states are structures over a one-element domain.
To see this, 	take an arbitrary domain. We construct equivalence classes of domain elements induced  by the subsets of the set of propositional variables.
	Each such  class is fully determined by a structure with one element. Thus, one domain element is enough.
	%	There are only finitely many of those, one for each possible subset of the propositional variables. 
	%\red{equivalence w.r.t. subformula-guarded bisimulation?}
%\end{proof}
 
%\begin{theorem}
%	In $L\mu\mu$, for the  propositional case, 
	Since we always have a fixed domain, temporal satisfiability is equivalent to the search version of temporal model checking, \TSAT= \TMC-{\rm SEARCH}, which, in turn, is related to model expansion in the classical logic setting. 
%\end{theorem}
Thus, the existence of a fixed domain, together with the force of information propagation,   explains robust decidability of modal logics that are built on top of propositional ones.
%\begin{proof}
%By the definition of the temporal satisfiability task, we need to find a valuation ${\cal V}$ .......................
%\end{proof}
With our explanation, all the previous explanations remain true. For instance, 
we do have only unary and binary predicates, but in a generalized sense.
Our guards are generalized too, and so is the tree model property.

\ignore{

 \subsection{Time Crystals}
 
 It seems that breaking the symmetry of time gives modal logics extra ``energy'' that is responsible for their robust decidability for satisfiability, compared to their ``flat'' counterparts - FO, MSO, FO(LFP).
 So, perhaps modal logics are a logician's version of time crystals? 
 
 The next question is in which case modal logic does not have a computational advantage over the ``flat'' counterpart for the satisfiability question?
 
 {\tt define a reversible fragment of the calculus. No energy is wasted in each gate. The crystal structure and modal advantages for satisfiability are lost}
 
}

\section{Conclusion}\label{sec:Conclusion}

The relational data model has recently been heavily criticized
for being outdated, not being able to  link data, to the degree that it's 
been called a ``legacy technology''. We have shown that those claims are perhaps premature.
The same data model that comes from Codd's relational algebra, with two modification that do not affect the main language, 
can represent relational data, graph databases and operations on them, powerful inter-connected solvers, 
ontologies, all in the same framework. 

While logic is already responsible for great technological advances, 
we strongly believe that the next 
technological shift can only happen if we change the units we operate on, as  we proposed here.

\bibliographystyle{plain}
\bibliography{Files/bibliography}

\end{document}